\documentclass[11pt]{article}
\usepackage{amssymb}
\usepackage{latexsym}
\usepackage{amsthm}
\usepackage{amsmath}
\usepackage[utf8]{inputenc}
\usepackage{fullpage}
\usepackage{boxedminipage}
\usepackage{listings}
\usepackage{minitoc}
\usepackage[pdftex]{graphicx}
\usepackage{graphicx}

\numberwithin{equation}{section}

\newcommand{\diag}{{\rm diag}}

\newcommand{\be}{\begin{equation}}
\newcommand{\ee}{\end{equation}}
\newcommand{\bes}{\begin{equation*}}
\newcommand{\ees}{\end{equation*}}
\newcommand{\eqn}{\begin{eqnarray}}
\newcommand{\feqn}{\end{eqnarray}}
\newcommand{\eqnn}{\begin{eqnarray*}}
\newcommand{\feqnn}{\end{eqnarray*}}
\newtheorem{theorem}{Theorem}

\newtheorem{prop}{Proposition}
\makeatletter \@addtoreset{equation}{section} \makeatother

\newif\ifpdf \ifx\pdfoutput\undefined \pdffalse
\else \pdfoutput=1 \pdftrue \fi \ifpdf \else \fi
\input epsf

\begin{document}

\ifpdf\DeclareGraphicsExtensions{.pdf, .jpg, .tif} \else%
\DeclareGraphicsExtensions{.eps, .jpg} \fi
\begin{titlepage}

    \thispagestyle{empty}
    \begin{flushright}
        \hfill{SU-ITP-10/10}\\
    \end{flushright}

    \vspace{30pt}
    \begin{center}
        { \Huge{\bf Iwasawa $\mathcal{%
N}=8$ Attractors}}

        \vspace{30pt}

        {\large{\bf Sergio L. Cacciatori$^\clubsuit$, Bianca L. Cerchiai$^{\diamondsuit}$, and \ Alessio Marrani$^{\heartsuit}$}}

        \vspace{40pt}

        {$\clubsuit$ \it Dipartimento di Fisica e Matematica,\\Universit\`a degli Studi dell'Insubria,
Via Valleggio 11, 22100 Como, Italy\\
and INFN, Sezione di Milano, Via Celoria 16, 20133 Milano, Italy
\texttt{sergio.cacciatori@uninsubria.it}}

        \vspace{5pt}

        {$\diamondsuit$ \it Dipartimento di Matematica,\\
Universit\`a degli Studi di Milano,  Via Saldini 50, 20133 Milano,
Italy\\
\texttt{bianca.cerchiai@unimi.it}}

        \vspace{5pt}

        {$\heartsuit$ \it Stanford Institute for Theoretical Physics\\
        Department of Physics, 382 Via Pueblo Mall, Varian Lab,\\
        Stanford University, Stanford, CA 94305-4060, USA\\
        \texttt{marrani@lnf.infn.it}}

        \vspace{55pt}
\end{center}

\vspace{5pt}

\begin{abstract}
Starting from the symplectic construction of the Lie algebra $\frak{e}%
_{7\left( 7\right) }$ due to Adams, we consider an Iwasawa
parametrization of the coset $\frac{E_{7\left( 7\right) }}{SU\left(
8\right) }$, which is the scalar manifold of $\mathcal{N}=8$, $d=4$
supergravity. Our approach, and the manifest \textit{off-shell}
symmetry of the resulting symplectic frame, is determined by a
non-compact Cartan subalgebra of the maximal subgroup
$SL(8,\mathbb{R})$ of $E_{7\left( 7\right) }$.

In absence of gauging, we utilize the explicit expression of the Lie algebra to study the origin of $%
\frac{E_{7\left( 7\right) }}{SU\left( 8\right) }$ as scalar
configuration of a $\frac{1}{8}$-BPS extremal black hole attractor.
In such a framework, we highlight the action of a $U\left( 1\right) $ symmetry spanning the dyonic $%
\frac{1}{8}$-BPS attractors. Within a suitable supersymmetry
truncation allowing for the embedding of the Reissner-N\"{o}rdstrom
black hole, this $U\left( 1\right)$ is interpreted as nothing but the global $\mathcal{R}$%
-symmetry of \textit{pure} $\mathcal{N}=2$ supergravity.

Moreover, we find that the above mentioned $U(1)$ symmetry is broken
down to a discrete subgroup $\mathbb{Z}_{4}$, implying that all
$\frac{1}{8}$-BPS Iwasawa attractors are non-dyonic near the origin
of the scalar manifold. We can trace this phenomenon back to the
fact that the Cartan subalgebra of $SL(8,\mathbb{R})$  used in our
construction endows the symplectic frame with a manifest
\textit{off-shell} covariance which is smaller than
$SL(8,\mathbb{R})$ itself. Thus, the consistence of the
Adams-Iwasawa symplectic basis with the
action of the $U\left( 1\right) $ symmetry gives rise to the observed $%
\mathbb{Z}_{4}$ residual non-dyonic symmetry.
\end{abstract}

\end{titlepage}
\newpage\tableofcontents

\section{\label{Intro}Introduction}

Local supersymmetry with $\mathcal{N}=8$ supercharge spinor generators is
the maximal one realized by a Lagrangian field theory with spin $s\leqslant
2 $ in $d=4$ space-time dimensions \cite{CJ,dWN}. No matter coupling is
allowed, and the bosonic content of the unique gravity supermultiplet is
given, besides the \textit{Vielbein}, by $28$ Abelian vector fields and $70$
real scalar fields. These latter coordinatize the symmetric coset
\begin{equation}
M_{\mathcal{N}=8,d=4}=\frac{E_{7\left( 7\right) }}{SU\left( 8\right) },
\label{scalar-manifold}
\end{equation}
where $E_{7\left( 7\right) }$ is the $U$-duality group \cite{HT} and $%
SU\left( 8\right) $ is its maximal compact subgroup. The $70$ real scalar
fields $\phi _{\left[ ijkl\right] }$ sit in the rank-$4$ completely
antisymmetric irrepr. $\mathbf{70}$ of $SU\left( 8\right) $ ($i,j=1,..,8$,
in the fund. irrepr. $\mathbf{8}$ of $SU(8)$). On the other hand, the
two-form Maxwell field strengths and their duals carry a symplectic index $%
\mathbb{A}$ sitting in the fundamental irrepr. $\mathbf{56}$ of $E_{7\left(
7\right) }$, which define the symplectic embedding of the $U$-duality
through the Gaillard-Zumino procedure \cite{GZ} (see also \textit{e.g.} \cite
{ADF-U-duality-revisited}). Thus, the fluxes of the two-form Maxwell field
strengths define the dyonic charge vector $\mathbf{Q}^{\mathbb{A}}$, which
then splits into electric and magnetic charges in a manifestly $SU\left(
8\right) $-covariant fashion as follows:
\begin{equation}
\mathbf{Q}^{\mathbb{A}}=\left( q_{ij},p^{ij}\right) ,  \label{charge-vector}
\end{equation}
where antisymmetrization is understood in the pairs of $SU\left( 8\right) $%
-indices.

$\mathcal{N}=8$ supersymmetry constrains the theory in a remarkably peculiar
way, which recently turned out to exhibit exceptional features. Indeed,
apart from being studied as a candidate for the simplest quantum field theory
\cite{simplest}, $\mathcal{N}=8$, $d=4$ supergravity has been shown to have
unexpected convergent ultraviolet properties, explicitly computed until four
loops in perturbation theory \cite{UV-finite-until-4-loops-incl}.

In absence of gauging, asymptotically flat, static, spherically symmetric,
dyonic, extremal (\textit{i.e.} zero temperature) black holes (BHs), with
various degrees of BPS-saturation, emerge as classical solutions of the
non-linear Einstein equations. According to \cite{GT}, these BHs can be seen as smooth solitonic $p=0$-branes, interpolating between two maximally supersymmetric $d=4$ geometries, namely Minkowski space at spatial infinity and conformally flat $AdS_{2}\times S^{2}$ Bertotti-Robinson \cite{BR} near-horizon geometry.

At spatial infinity, such BHs are characterized by their ADM mass \cite{ADM}%
, depending both on $\mathbf{Q}^{\mathbb{A}}$ and on the asymptotical
unconstrained values $\phi _{\left[ ijkl\right] \infty }$ of the scalar
fields. The area $A_{H}$ of the BH event horizon, and thus, through the
Bekenstein-Hawking formula \cite{BH}, the BH entropy $S_{BH}$, is given
purely in terms $\mathbf{Q}^{\mathbb{A}}$, thanks to the \textit{Attractor
Mechanism} \cite{FKS}-\nocite{Strom,FK1,FK2}\cite{FGK} \cite{Kol}:
\begin{equation}
\frac{S_{BH}}{\pi }=\frac{A_{H}}{4\pi }=\sqrt{\left| \mathcal{I}_{4}\right| }%
,
\end{equation}
where $\mathcal{I}_{4}$ is the unique quartic Cartan invariant \cite{I_4} of
$E_{7\left( 7\right) }$, defined in terms of the rank-$4$ completely
symmetric invariant tensor $K_{\left( \mathbb{ABCD}\right) }$ of the $%
\mathbf{56}$ of $E_{7\left( 7\right) }$ as follows:
\begin{equation}
\mathcal{I}_{4}\equiv K_{\mathbb{ABCD}}\mathbf{Q}^{\mathbb{A}}\mathbf{Q}^{%
\mathbb{B}}\mathbf{Q}^{\mathbb{C}}\mathbf{Q}^{\mathbb{D}}.
\end{equation}

Following the general analysis \cite{Ferrara-Maldacena,FG,LPS-1,ADFL-0-brane} of
$E_{7\left( 7\right) }$-invariant BPS conditions for the various classes of
BH states, as well as of the corresponding charge orbits of the $\mathbf{56}$
of $E_{7\left( 7\right) }$, extremal BH attractors in $\mathcal{N}=8$, $d=4$
supergravity were studied in \cite{FK-N=8,Gimon} (see also \cite{ADFFT,ACC},
as well as the recent treatment in \cite{CFMZ}), by solving the criticality
conditions \cite{FGK} for the effective BH potential
\begin{equation}
V_{BH}\equiv \frac{1}{2}Z_{ij}\overline{Z}^{ij},  \label{V_BH-def}
\end{equation}
where $Z_{ij}$ is the $\mathcal{N}=8$, $d=4$ complex antisymmetric central
charge matrix (see \textit{e.g.} \cite
{Ferrara-Maldacena,ADF-U-duality-revisited}, and Refs. therein). Then, in
\cite{CFGM1} some simple configurations were considered, corresponding to
the well-known typologies of Reissner-N\"{o}rdstrom, Kaluza-Klein and
axion-dilaton BHs. Through suitable branching decompositions of the relevant
irreprs. of the $U$-duality group $E_{7\left( 7\right) }$, these well-known
solutions were shown to be embedded in maximal $d=4$ supergravity. Such an
analysis has been further developed in \cite{CFGn-1}, where the relations
between extremal $d=4$ BHs and extremal $d=5$ BHs and black strings have
been studied, by exploiting the connection between $E_{7\left( 7\right) }$
and the $d=5$ $U$-duality group $E_{6\left( 6\right) }$. To this aim, $%
\mathcal{N}=8$, $d=4$ supergravity has been formulated in a manifestly $%
E_{6\left( 6\right) }$-covariant basis \cite{ADFL-gauging-flat}, namely the
one related to the Sezgin-Van Nieuwenhuizen $d=5\rightarrow d=4$ dimensional
reduction \cite{SVN}. This is not the same as the Cremmer-Julia \cite{CJ} or
de Wit-Nicolai \cite{dWN} symplectic frame, whose maximal non-compact
\textit{off-shell} symmetry is $SL\left( 8,\mathbb{R}\right) $. The relation
between these two formulations, usually adopted to study $d=4$ maximal
supergravity in absence of gauging, has been precisely discussed in \cite
{CFGM1}, and it amounts to dualizing several vector fields and therefore to
interchanging the electric and magnetic charges of some of the $28$ Abelian
vector fields of the theory.

Furthermore, extremal BH attractors provide an interesting arena, in which
the above mentioned issues of ultraviolet convergence of perturbative
quantum field theory computations (leading to the conjecture of ultra-violet
finiteness of $\mathcal{N}=8$, $d=4$ supergravity) have been recently
investigated (see \cite{BFK-1,BFK-2}, and Refs. therein; see also \cite
{Sen-Arithmetic-N=8,ICL-1}).\medskip

\textit{\c{C}a va sans dire}, the Sezgin-Van Nieuwenhuizen \cite{SVN} and
Cremmer-Julia \cite{CJ} or de Wit-Nicolai \cite{dWN} symplectic frames are
not the only ones in which $\mathcal{N}=8$, $d=4$ supergravity can be
formulated. Apart from the bases related to the various possible gaugings of
the theory (see \textit{e.g.} \cite{Duff,HW,Gauging-TO-1}, and Refs.
therein), other ungauged formulations can be considered, and they can be
useful to unveil some interesting facets of the theory itself.

Hinted by Adam's approach to the Lie algebra $\frak{e}_{7\left( 7\right) }$
\cite{adams}, in this paper we explicitly perform an Iwasawa parametrization
of the coset representative of the symmetric manifold $\frac{E_{7\left(
7\right) }}{SU\left( 8\right) }$. The main feature of such a construction is
the use of a completely non-compact $7$-dimensional Cartan subalgebra of $%
SL\left( 8,\mathbb{R}\right) $, which leads to the nilpotency of the matrix
realization of the relevant coset generators, determining the maximal
manifest covariance of the whole framework to be $%
SL\left( 7,\mathbb{R}\right) $. Considering the expression of the
coset (\ref{scalar-manifold}) to the first order, i.e. at the Lie algebra level, we then study the $SU\left( 8\right) $%
-invariant origin of such a manifold as a $\frac{1}{8}$-BPS scalar
configuration corresponding to an extremal BH attractor \cite{FK-N=8,CFGM1}.

Within such an approach to $\frac{1}{8}$-BPS attractors, we remark the
existence of a \textit{residual ``degeneracy symmetry''} $U\left( 1\right) $%
. This symmetry is \textit{residual}, because it characterizes the (particular
representative of the) orbit of charge configurations which support $\frac{1}{8}$-BPS attractors. Furthermore, it is a \textit{%
``degeneracy symmetry''} because it spans the dyonic nature of the $\frac{1}{%
8}$-BPS solutions exhibiting an \textit{Attractor Mechanism}. As also
pointed in the analysis of \cite{CFGM1}, this symmetry is decompactified
to $SO\left( 1,1\right) $ in non-BPS attractors, thus not allowing for the
origin of $\frac{E_{7\left( 7\right) }}{SU\left( 8\right) }$ to constitute a
representative of non-BPS $\mathcal{N}=8$ attractor scalar configurations.

The main result of the present investigation is the discovery that such a $%
U\left( 1\right) $ symmetry, characterizing the $\frac{1}{8}$-BPS attractors
in $\frac{E_{7\left( 7\right) }}{SU\left( 8\right) }$,
actually gets spoiled within the coset construction \textit{\`{a} la
Adams-Iwasawa} performed in the paper. Indeed, $U\left( 1\right) $ is
broken down to a discrete $\mathbb{Z}_{4}$ subgroup, as it appears from the
\textit{purely electric} or \textit{purely magnetic} nature of the solutions
to the set of Attractor Equations governing the near-horizon dynamics of the
scalar fields. By analyzing such a $U\left( 1\right) \rightarrow \mathbb{Z}%
_{4}$ breaking in detail, we are able to trace its origin back to the
maximal manifest \textit{off-shell} covariance properties of the construction, i.e., to the choice of a $7$-dimensional completely non-compact Cartan subalgebra of $%
SL\left( 8,\mathbb{R}\right) $, which breaks the maximal manifest \textit{%
off-shell} covariance down to $SL\left( 7,\mathbb{R}\right) $, or, through a
suitable Cayley rotation, to $SU\left( 7\right) $.

Thus, our investigation points out that the dyonic nature of $\frac{1}{8}$%
-BPS extremal BH attractor in \textit{ungauged} $\mathcal{N}=8$, $d=4$
supergravity essentially relies on the covariance properties exhibited by
the parametrization chosen for the scalar manifold (\ref{scalar-manifold})
itself. As explained in the concluding Sect. \ref{Conclusion}, each of the
symplectic frames mentioned above is ``natural'' in order to explicit the
maximal symmetry of a class of attractors. In this perspective, the
Lie-algebra approach to the Adams-Iwasawa construction of $\frac{E_{7\left( 7\right) }}{%
SU\left( 8\right) }$ studied in the present paper highlights the action of a
$U\left( 1\right) $ symmetry in the dyonic attractor solutions pertaining to
$\frac{1}{8}$-BPS BH states, and its breaking to a discrete subgroup.

It should be remarked that, in light of the embedding analysis performed in
\cite{CFGM1}, the Lie algebra limit of $\frac{1}{8}$-BPS attractors is related to the embedding of the Reissner-N\"{o}rdstrom extremal BH solution of \textit{%
pure} $\mathcal{N}=2$, $d=4$ supergravity into a $\mathcal{N}=8$ maximal
theory. In this framework, the \textit{residual} $U(1)$ \textit{``degeneracy
symmetry''} mentioned above is nothing but the $U\left( 1\right) $ global $%
\mathcal{R}$-symmetry of \textit{pure} $\mathcal{N}=2$, $d=4$ theory%
\footnote{%
Indeed, in absence of scalars the $\mathcal{R}$-symmetry, usually contained
in the stabilizer of the scalar manifold, gets promoted to a global ($U$%
-duality) symmetry \cite{FSZ}.} \cite{GH}. In this context, the
$U\left( 1\right) \rightarrow \mathbb{Z}_{4}$ breaking due to the
constraints on manifest covariance imposed by the Adams-Iwasawa construction, can be interpreted as a breaking of the $%
\mathcal{R}$-symmetry of the \textit{pure} $\mathcal{N}=2$, $d=4$ theory, to
which $\mathcal{N}=8$, $d=4$ supergravity gets effectively truncated in the
sector of $\frac{1}{8}$-BPS attractors near the origin of (\ref
{scalar-manifold}).\bigskip\

The plan of the paper is as follows.

Sect. \ref{Adams-Iwasawa} is devoted to a detailed construction of the coset
representative of the $\mathcal{N}=8$, $d=4$ scalar manifold (\ref
{scalar-manifold}), by exploiting Adams' realization \cite{adams} of the Lie
algebra of the $U$-duality group $E_{7\left( 7\right) }$ (Subsect. \ref
{Lie-Algebra}). Then, using a completely non-compact $7$-dimensional Cartan
subalgebra of $SL\left( 8,\mathbb{R}\right) $ as a pivot, in Subsect. \ref
{Iwasawa} a parametrization \textit{\`{a} la Iwasawa} of the coset
representative is worked out.

Then, Sect. \ref{N=8,d=4-sugra} deals with the formulation of $\mathcal{N}=8$%
, $d=4$ ungauged supergravity theory within such a symplectic frame,
computing the central charge matrix $Z_{ij}$ and the effective BH potential $%
V_{BH}$ in terms of the symplectic electric and magnetic sections.

At the Lie algebra level, exploring the \textit{Attractor Mechanism} in the
neighbourhood of the origin of the scalar manifold (\ref{scalar-manifold})
itself, the $\frac{1}{8}$-BPS attractor solutions are studied in Sect. \ref
{0=1/8-BPS-large}. {From} the analysis performed in Subsect. \ref
{Iwa}, only purely electric or purely magnetic Iwasawa solutions are
obtained. The \textit{non-dyonic} nature of such attractors is then
investigated in Subsect. \ref{Breaking}, in which it is found that such a
phenomenon is due to the breaking of the \textit{residual ``degeneracy
symmetry''} $U\left( 1\right) $ down to a subgroup $\mathbb{Z}_{4}$.

Concerning the $d=5$ uplift properties of maximal $d=4$ supergravity, Sect.
\ref{d=4-d=5} reports some relations between scalar manifolds and \textit{%
moduli spaces} of attractors, with some new observations related to the $c$%
-map \cite{CFG} and thus to $d=3$ (non-maximal) theories, hinting to
further developments, also in view of recent advances in the field (see
\textit{e.g.} \cite{d=3-recent}).

Some final remarks and comments are given in concluding Sect. \ref
{Conclusion}.

\section{\label{Adams-Iwasawa}Adams-Iwasawa Approach to $\frak{e}_{7(7)}/%
\frak{su}(8)$}

\subsection{\label{Lie-Algebra}Adams' Symplectic Construction of $\frak{e}%
_{7(7)}$}

In order to obtain a direct construction of the maximally non-compact
exceptional Lie algebra $\frak{e}_{7(7)}$, we follow Chapter 12 of \cite
{adams}.

Let $V$ be an $8$-dimensional real vector space and $V^{\ast }$ its dual.
The notation $\Lambda ^{i}V$ denotes the $i$-th external power of $V$. By
exploiting the isomorphism $\Lambda ^{8}V\simeq \mathbb{R}$, one can then
define $SL(V)$ as the group of automorphisms preserving such isomorphism.
Then, the Lie algebra of $SL(V)$ itself can be defined:
\begin{equation}
L\equiv \frak{sl}(V).  \label{L-def}
\end{equation}
$L$ acts on the $56$-dimensional real vector space
\begin{equation}
W\equiv \Lambda ^{2}V\oplus \Lambda ^{2}V^{\ast }  \label{W-decomp}
\end{equation}
in the usual way, namely:
\begin{equation}
L(W)=L(V)\wedge V\oplus L(V^{\ast })\wedge V^{\ast }+V\wedge L(V)\oplus
V^{\ast }\wedge L(V^{\ast }),  \label{L-decomp-1}
\end{equation}
where $L(V^{\ast })$ denotes the adjoint action.\newline
If $i+j=8$, the pairing
\begin{equation}
\Lambda ^{i}V\otimes \Lambda ^{j}V\longrightarrow \Lambda ^{8}V\simeq
\mathbb{R},  \label{pairing}
\end{equation}
given by the wedge product $\wedge $, defines an isomorphism
\begin{equation}
\Lambda ^{i}V\simeq \Lambda ^{j}V^{\ast }.  \label{iso-1}
\end{equation}
Such an isomorphism can then be used to define an action of
\begin{equation}
\lambda ^{4}\equiv \Lambda ^{4}V  \label{lambda^4}
\end{equation}
(with $dim_{\mathbb{R}}=\binom{8}{4}=70$) on $W$ by means of the maps:
\begin{eqnarray}
&&\lambda ^{4}\otimes \Lambda ^{2}V\overset{\wedge }{\longrightarrow }%
\Lambda ^{6}V\simeq \Lambda ^{2}V^{\ast };  \label{actionV} \\
&&  \notag \\
&&\lambda ^{4}\otimes \Lambda ^{2}V^{\ast }\simeq \Lambda ^{4}V^{\ast
}\otimes \Lambda ^{2}V^{\ast }\overset{\wedge }{\longrightarrow }\Lambda
^{6}V^{\ast }\simeq \Lambda ^{2}V.  \label{actionV*}
\end{eqnarray}
Thus, it follows that
\begin{equation}
A\equiv L\oplus \lambda ^{4}  \label{A-deff}
\end{equation}
is a $133$-dimensional real vector space of operators acting on $W$.

The following Theorem holds (cfr. Theorem 12.1, as well as the end of
Chapter 12, of \cite{adams}):

\begin{theorem}
$A$ is a Lie algebra of maps which acts on $W$ in the same way as $\frak{e}%
_{7(7)}$ acts on its fundamental irrepr. $\mathbf{56}$, up to isomorphisms.
\end{theorem}

Then, $A$ is a realization of the Lie algebra $\frak{e}_{7(7)}$ with
irreducible representation $(A,W)$. Up to isomorphisms, this is indeed the
smallest faithful representation of $\frak{e}_{7(7)}$.


\subsubsection{\label{Matrix}Matrix Realization}

After identifying $V$ with $\mathbb{R}^{8}$, the action of $L$ on $V$ is
generated by the action of the traceless $8\times 8$ matrices in $M(8,%
\mathbb{R})$. One can then choose a basis $\{e_{i}\}_{i=1,...,8}$ of $V$ and
a basis $\{A_{kl},S_{kl},D_{\alpha }\}$ (with cardinality $8\times 8-1=63$)
for $M(8,\mathbb{R})$, defined as follows\footnote{%
Throughout the whole treatment, as usual, the square brackets denote
antisymmetrization of enclosed indices, according to the definition $%
A_{[kl]}\equiv \frac{1}{2}(A_{kl}-A_{lk})$, while the round brackets
indicate symmetrization of enclosed indices: $S_{(kl)}\equiv \frac{1}{2}%
(S_{kl}+S_{lk})$.
\par
It is worth remarking that often the symmetry properties are used to
introduce an ordering rule, and to restrict the range of the indices.} ($%
1\leq k<l\leq 8$, and $\alpha =1,\ldots ,7$):
\begin{eqnarray}
A_{kl}e_{i} &\equiv &\delta _{li}e_{k}-\delta _{ki}e_{l}=A_{[kl]}e_{i};
\label{A-antisymm-def} \\
S_{kl}e_{i} &\equiv &\delta _{li}e_{k}+\delta _{ki}e_{l}=S_{(kl)}e_{i};
\label{S-symm-def} \\
D_{\alpha } &\equiv &\text{diag}\{D_{\alpha }^{1},\ldots ,D_{\alpha
}^{8}\};~~Tr\left( D_{\alpha }\right) =0.  \label{D-Tr}
\end{eqnarray}
Thus, $A_{kl}$'s and $S_{kl}$'s respectively are $28$ antisymmetric and $28$
symmetric\footnote{%
Notice that, despite their traceless symmetry, $S_{kl}$'s are only $28$ (and
not $35$), because the index ordering $k<l$ has been enforced. The $8-1=7$
traceless diagonal (\textit{i.e.} $k=l$) degrees of freedom of $S_{kl}$'s
are implemented through the $D_{\alpha }$'s.} $8\times 8$ matrices, whereas $%
D_{\alpha }$'s are $7$ diagonal traceless $8\times 8$ matrices, which can be
identified with the Cartan subalgebra of $\frak{e}_{7(7)}$ (see further
below). Their normalization is chosen such that
\begin{eqnarray}
Tr\left( A_{kl}A_{mn}\right) &\equiv &-2\delta _{kl\mid mn}; \\
Tr\left( S_{kl}S_{mn}\right) &\equiv &2\delta _{kl\mid mn}; \\
Tr\left( D_{\alpha }D_{\beta }\right) &\equiv &2\delta _{\alpha \beta }.
\end{eqnarray}

It is now possible to extend the action of $L$ on $V$ (defined above) to $%
\Lambda ^{2}V$, and next to $W$. To this aim, let us introduce $\left\{
e_{ij}\right\} _{i<j}\equiv e_{i}\wedge e_{j}$ as basis for $\Lambda ^{2}V$,
and denote its dual by $\{\varepsilon ^{ij}\}_{i<j}$. Thus, one reaches the
following results\footnote{%
In these sums we do not restrict $m<n$, rather we take into account that $%
e_{mn}=-e_{nm}$. It similarly holds for Eqs. (\ref{1})-(\ref{3}).}:
\begin{eqnarray}
&&A_{kl}(e_{ij})=\sum_{m,n}(U_{klim}^{A}D_{kljn}+D_{klim}U_{kljn}^{A})e_{mn};
\label{PPPA-1} \\
&&S_{kl}(e_{ij})=\sum_{m,n}(U_{klim}^{S}D_{kljn}+D_{klim}U_{kljn}^{S})e_{mn};
\label{PPPA-2} \\
&&D_{\alpha }(e_{ij})=(D_{\alpha }^{i}+D_{\alpha }^{j})e_{ij},
\label{PPPA-3}
\end{eqnarray}
where the quantities
\begin{eqnarray}
U_{klim}^{A} &\equiv &\delta _{km}\delta _{li}-\delta _{ki}\delta _{lm};
\label{PPPA-4} \\
U_{klim}^{S} &\equiv &\delta _{km}\delta _{li}+\delta _{ki}\delta _{lm}; \\
D_{klim} &\equiv &\left\{
\begin{array}{l}
\delta _{im}~\text{for}~k\neq l\neq i; \\
0~\text{otherwise}.
\end{array}
\right. ~  \label{PPPA-6}
\end{eqnarray}
have been introduced. On the dual basis $\varepsilon ^{ij}$, $M(8,\mathbb{R}%
) $ acts as minus the transposed matrix:
\begin{eqnarray}
&&A_{kl}(\varepsilon
^{ij})=\sum_{m,n}(U_{klim}^{A}D_{kljn}+D_{klim}U_{kljn}^{A})\varepsilon
^{mn};  \label{1} \\
&&S_{kl}(\varepsilon
^{ij})=-\sum_{m,n}(U_{klim}^{S}D_{kljn}+D_{klim}U_{kljn}^{S})\varepsilon
^{mn};  \label{2} \\
&&D_{\alpha }(\varepsilon ^{ij})=-(D_{\alpha }^{i}+D_{\alpha
}^{j})\varepsilon ^{ij}.  \label{3}
\end{eqnarray}
Equations (\ref{1}), (\ref{PPPA-1}) and (\ref{2}), (\ref{PPPA-2}) define the
$56 \times 56$ matrices representing the action on $W$ of the operators $%
A_{kl}$ and $S_{kl}$ respectively. We will keep the names $A_{kl}$ and $%
S_{kl}$ for such matrices. Similarly, (\ref{3}) and (\ref{PPPA-3}) define
the $56 \times 56$ matrices $h_{D_\alpha}$ corresponding to the diagonal
matrices $D_\alpha$.

In order to determine the remaining $70$ generators of $\frak{e}_{7(7)}$
(which span $\lambda ^{4}$ defined by (\ref{lambda^4})), we consider the
action of
\begin{equation}
\underset{\left( i_{1}<i_{2}<i_{3}<i_{4}\right) }{\lambda
_{i_{1}i_{2}i_{3}i_{4}}}\equiv e_{i_{1}}\wedge e_{i_{2}}\wedge
e_{i_{3}}\wedge e_{i_{4}}
\end{equation}
on $W$. By exploiting the identifications (\ref{actionV}) and (\ref{actionV*}%
), this yields to
\begin{eqnarray}
&&(e_{i_{1}}\wedge e_{i_{2}}\wedge e_{i_{3}}\wedge e_{i_{4}})\otimes
(e_{j_{1}j_{2}})\mapsto \frac{1}{2}\epsilon
_{i_{1}i_{2}i_{3}i_{4}j_{1}j_{2}k_{1}k_{2}}\varepsilon ^{k_{1}k_{2}}; \\
&&(e_{i_{1}}\wedge e_{i_{2}}\wedge e_{i_{3}}\wedge e_{i_{4}})\otimes
(\varepsilon ^{j_{1}j_{2}})\mapsto \frac{1}{2}\delta
_{i_{1}i_{2}i_{3}i_{4}}^{j_{1}j_{2}k_{1}k_{2}}e_{k_{1}k_{2}},
\end{eqnarray}
where $\epsilon $ is the standard $8$-dimensional Levi-Civita tensor.
Furthermore
\begin{equation}
\delta _{i_{1}i_{2}i_{3}i_{4}}^{j_{1}j_{2}j_{3}j_{4}}\equiv \sum_{\sigma \in
{\mathcal{P}}[1,2,3,4]}\epsilon _{\sigma }\delta _{i_{\sigma
(1)}}^{j_{1}}\delta _{i_{\sigma (2)}}^{j_{2}}\delta _{i_{\sigma
(3)}}^{j_{3}}\delta _{i_{\sigma (4)}}^{j_{4}},
\end{equation}
where ${\mathcal{P}}[1,2,3,4]$ denotes the set of permutations of $[1,2,3,4]$%
, and $\epsilon _{\sigma }$ stands for the parity of permutation $\sigma $.

Within the basis $e\equiv e_{ij}$ ($i<j$), it is convenient to use a
double-index notation for matrices, such that \textit{e.g.} the action of
the matrix $M$ on $e$ reads:
\begin{equation}
(Me)^{ij}=\sum_{k<l}M^{ij\mid kl}e_{kl}.
\end{equation}
Thus, the action of $\lambda _{i_{1}i_{2}i_{3}i_{4}}$, written in block
matrix form with respect to the decomposition (\ref{W-decomp}), reads%
\footnote{%
In (\ref{MI-1}) subscripts ``$u$'' and ``$d$'' simply stand for ``up''
respectively ``down'', referring to the position of the block matrices
within $\lambda _{i_{1}i_{2}i_{3}i_{4}}$ (see also Eq. (\ref{MI-2})).}
\begin{equation}
\lambda _{i_{1}i_{2}i_{3}i_{4}}=\left(
\begin{array}{cc}
0 & \lambda _{u}(i_{1},i_{2},i_{3},i_{4})_{ij\mid kl} \\
\lambda _{d}(i_{1},i_{2},i_{3},i_{4})^{ij\mid kl} & 0
\end{array}
\right) \equiv \left(
\begin{array}{cc}
0 & \epsilon _{i_{1}i_{2}i_{3}i_{4}ijkl} \\
\delta _{i_{1}i_{2}i_{3}i_{4}}^{ijkl} & 0
\end{array}
\right) ,  \label{MI-1}
\end{equation}
where it is worth pointing out that the matrices $\lambda _{u}$ and $\lambda
_{d}$ are both symmetric.

It is now convenient to introduce the tetra-indices $I\equiv \left[
i_{1}i_{2}i_{3}i_{4}\right] $ (notice the complete antisymmetrization),
endowed with the ordering rule $i_{1}<i_{2}<i_{3}<i_{4}$. This in turn
uniquely determines the \textit{complementary} tetra-index $\tilde{I}$, such
that $\epsilon _{I\tilde{I}}\neq 0$. As a consequence, it follows that
\begin{equation}
\left( \lambda _{I}\right) ^{T}=\epsilon _{I\tilde{I}}\lambda _{\tilde{I}}.
\end{equation}
This allows for a change of basis in $\lambda ^{4}$ through the introduction
of the symmetric matrices
\begin{equation}
{}\mathcal{S}_{I}\equiv \frac{1}{\sqrt{2}}(\lambda _{I}+\epsilon _{I\tilde{I}%
}\lambda _{\tilde{I}}),  \label{S-def}
\end{equation}
as well as of antisymmetric matrices
\begin{equation}
\mathcal{A}{}_{I}\equiv \frac{1}{\sqrt{2}}(\lambda _{I}-\epsilon _{I\tilde{I}%
}\lambda _{\tilde{I}}).  \label{A-def}
\end{equation}

Since each of the sets of tetra-indices $\mathcal{I}\equiv \left\{ I\right\}
$ and $\widetilde{\mathcal{I}}\equiv \left\{ \widetilde{I}\right\} $ has
cardinality $70$, the definitions (\ref{S-def}) and (\ref{A-def}) exhibit a
double over-counting, namely only half of the ${}\mathcal{S}_{I}$'s and of
the $\mathcal{A}{}_{I}$'s is linearly independent. Indeed, it holds that ${}%
\mathcal{S}_{I}={}\mathcal{S}_{\widetilde{I}}$ and $\mathcal{A}{}_{I}=-%
\mathcal{A}{}_{\widetilde{I}}$. In order to restrict the set of
tetra-indices to a consistent basis, it should be noted that the subset $%
\mathcal{I}_{8}$ of tetra-indices $\mathcal{I}$ containing\footnote{%
Note that one might have instead fixed one index out of $\left\{
i_{1},i_{2},i_{3},i_{4}\right\} $ to one of the values $\left\{
1,...,8\right\} $. All such choices are equivalent to fixing $i_{4}=8$, due
to the complete antisymmetrization of the four indices $i_{1}i_{2}i_{3}i_{4}$%
.} $8$ has cardinality $35$. Indeed, the index ordering determines the
unique independent configuration to have $i_{4}=8$, so that $\mathcal{I}%
_{8}\equiv \left[ ijk8\right] $, with $1\leq i<j<k\leq 7$. The cardinality
of $\mathcal{I}_{8}$ is thus $\binom{7}{3}=\frac{1}{2}\binom{8}{4}=35$.
Subsequently, the complement of $\mathcal{I}_{8}$ in $\mathcal{I}$ can be
defined as the set of $35$ tetra-indices in $\left\{ I\right\} $
complementary to the ones in $\mathcal{I}_{8}$. Therefore, one can conclude
that a basis for $\lambda ^{4}$ is consistently provided by the set
\begin{equation}
\{{}\mathcal{S}_{I},{}\mathcal{A}{}_{I}\}_{I\in {}\mathcal{I}_{8}},
\label{lambda^4-basis}
\end{equation}
with cardinality $35+35=70$.

Furthermore, through (\ref{A-deff}), a basis for $A$ is given by
\begin{eqnarray}
&&\{A_{kl},{}\mathcal{A}_{I},h_{D_{\alpha }},S_{kl},\mathcal{S}_{I}\},
\label{basis} \\
&&\left\{
\begin{array}{l}
1\leq k<l\leq 8; \\
1\leq \alpha \leq 8; \\
I\in \mathcal{I}_{8}.
\end{array}
\right.
\end{eqnarray}
It is worth noting that, by construction, all matrices in (\ref{basis}) are
orthogonal. The set of antisymmetric matrices $A_{\mu }\equiv \{A_{kl},%
\mathcal{A}_{I}\}$ is normalized as $Tr\left( A_{\mu }A_{\nu }\right)
=-2\delta _{\mu \nu }$, and it has cardinality $28+35=63$ ($\mu =1,...,63$),
so that $A_{\mu }$ generates the maximal compact (symmetric) subgroup $SU(8)$
of $E_{7\left( 7\right) }$ (see \textit{e.g.} \cite{Gilmore}). The remaining
set of $7+28+35=70$ symmetric generators $S_{\Lambda }\equiv \{h_{D_{\alpha
}},S_{kl},{}\mathcal{S}_{I}\}$ ($\Lambda =1,...,70$) is normalized as $%
Tr\left( S_{\Lambda }S_{M}\right) =2\delta _{\Lambda M}$, so that it spans
the non-compact part of $\frak{e}_{7(7)}$. In particular, as mentioned
above, the $7$ diagonal matrices $D_{\alpha }$ (or equivalently $%
h_{D_{\alpha }}$, see definition (\ref{D}) below) generate a Cartan
subalgebra\footnote{%
Throughout our treatment, $\left\langle \mathcal{A}\right\rangle _{\mathbb{F}%
}$ denotes the set of linear combinations of elements of $\mathcal{A}$ with
coefficients in the ground field $\mathbb{F}$.}
\begin{equation}
\mathcal{C}\equiv \left\langle D_{\alpha }\right\rangle _{\mathbb{R}%
}\subsetneq \frak{e}_{7(7)},  \label{C-def}
\end{equation}
containing no compact elements. Thus, the cardinality of the basis (\ref
{basis}) of $A$ is $63+70=133$, consistent with (\ref{A-deff}), (\ref
{A-antisymm-def}), (\ref{S-symm-def}), (\ref{D-Tr}) and (\ref{lambda^4-basis}%
). \medskip

We can now proceed to perform an explicit Iwasawa parametrization of the
algebra $\frak{e}_{7(7)}/\frak{su}(8)$ underlying the $70$-dim. real
symmetric coset manifold $E_{7(7)}/SU(8)$ (namely, the scalar manifold of $%
\mathcal{N}=8$, $d=4$ supergravity (\ref{scalar-manifold})).


\subsection{\label{Iwasawa}Iwasawa Parametrization of $\frak{e}_{7(7)}/\frak{%
su}(8)$}

As the first step, one needs to choose a complete set of positive roots with
respect to the Cartan subalgebra $\mathcal{C}$ defined by (\ref{C-def}). To
this aim, it is convenient to introduce the following notation (recall
definition (\ref{L-def})):
\begin{equation}
L=\mathcal{C}\oplus \left\langle J^{+}\right\rangle _{\mathbb{R}}\oplus
\left\langle J^{-}\right\rangle _{\mathbb{R}},
\end{equation}
where
\begin{eqnarray}
\left\langle J^{+}\right\rangle _{\mathbb{R}} &\equiv &\left\langle
\{J_{kl}^{+}\equiv \frac{1}{\sqrt{2}}(S_{kl}+A_{kl})\,|\,k<l\}\right\rangle
_{\mathbb{R}};  \label{PPP-eve-1} \\
\left\langle J^{-}\right\rangle _{\mathbb{R}} &\equiv &\left\langle
\{J_{kl}^{-}\equiv \frac{1}{\sqrt{2}}(S_{kl}-A_{kl})\,|\,k<l\}\right\rangle
_{\mathbb{R}}.
\end{eqnarray}
Furthermore, it holds that (recall (\ref{L-def}) and (\ref{lambda^4}), as
well as (\ref{MI-1}))
\begin{equation}
\lambda ^{4}=\left\langle \mathcal{J}\right\rangle _{\mathbb{R}},
\end{equation}
where
\begin{equation}
\left\langle \mathcal{J}\right\rangle _{\mathbb{R}}\equiv \left\langle \{{}%
\mathcal{J}_{I}\equiv \lambda _{I}|I\in \mathcal{I}{}\}\right\rangle _{%
\mathbb{R}}.  \label{J-call-def}
\end{equation}

It is then possible to prove the following

\begin{prop}
\label{prop1} The set $J^{+}\cup J^{-}\cup \mathcal{J}$ diagonalizes
simultaneously the adjoint action of $\mathcal{C}$.
\end{prop}

\begin{proof}
The proof proceeds by direct computation. The adjoint action of
$[h_{D_\alpha}, J^\pm_{kl}]$ on $e_{ij}$ and on $\varepsilon_{ij}$
turns out to be:
\begin{eqnarray}
[h_{D_\alpha}, J^\pm_{kl}]=\pm (D_\alpha^k-D_\alpha^l) J^\pm_{kl},
\end{eqnarray}
while for ${\mathcal J}_{I}=\lambda_{i_1i_2i_3i_4}$ (recall definition (\ref{J-call-def})) one can compute
that (recalling the tracelessness of $D_\alpha$'s, see Eq.
(\ref{D-Tr}))
\begin{eqnarray}
[h_{D_\alpha}, {\mathcal J}_I]=(D_\alpha^{i_1}+D_\alpha^{i_2}+D_\alpha^{i_3}+D_\alpha^{i_4}) {\mathcal J}_I.
\end{eqnarray}
\end{proof}

To proceed further, we split the set $\mathcal{J}$ defined in (\ref
{J-call-def}) as follows:
\begin{equation}
\mathcal{J}=\mathcal{J}^{+}\cup \mathcal{J}^{-},
\end{equation}
where
\begin{eqnarray}
\mathcal{J}^{+}\equiv \{\lambda _{I}\in &&\mathcal{J}{}:I\in \mathcal{I}%
_{8}\};  \label{J+-def} \\
\mathcal{J}^{-}\equiv \{\lambda _{I}\in &&\mathcal{J}{}:I\notin \mathcal{I}%
_{8}\}.
\end{eqnarray}
Moreover, by extending the notation $h_{D_{\alpha }}$ introduced above, in
the treatment below we will denote by $h_{D}$ the $56\times 56$ matrix
representation of any $8\times 8$ traceless diagonal matrix $D$ with real
entries (which is then a ``diagonal'' element of $\frak{sl}(8)$). Then, one
can prove the

\begin{prop}
\label{prop2} The set $J^{+}\cup \mathcal{J}$ $^{+}$ defines a choice of
positive roots of $\frak{e}_{7(7)}$. The corresponding roots are the
operators
\begin{equation}
\{\beta _{kl}\}_{k<l}\cup \{\beta
_{i_{1}i_{2}i_{3}i_{4}}\}_{i_{1}i_{2}i_{3}i_{4}\in {}\mathcal{I}_{8}},
\end{equation}
defined by
\begin{eqnarray}
&&\beta _{kl}(h_{D})\equiv D^{k}-D^{l},\qquad k<l;  \label{choice-1} \\
&&\beta _{i_{1}i_{2}i_{3}i_{4}}(h_{D})\equiv
D^{i_{1}}+D^{i_{2}}+D^{i_{3}}+D^{i_{4}},\qquad i_{1}i_{2}i_{3}i_{4}\in
\mathcal{I}_{8}.  \label{choice-2}
\end{eqnarray}
\end{prop}

\begin{proof}
The fact that the set (\ref{choice-1})-(\ref{choice-2}) is the subset of eigen-matrices for the
adjoint action of $\mathcal{C}$ corresponding to the set of given roots,
follows immediately from Proposition \ref{prop1}. One only needs to
check that such roots lie in a convex cone. To this end, it is
sufficient to show that there exists at least one matrix $c\in
\mathcal{C}$ such that $\beta_{k<l}(c)>0$ and $\beta_{{\mathcal
I}_8} (c)>0$. In fact, this is the case if \textit{e.g.} the
following $c$ is chosen:
\begin{equation}
c=\diag \{-1,-2,-3,-4,-5,-6,-7,28\}. \end{equation}
\end{proof}

Within the choices (\ref{choice-1}) and (\ref{choice-2}) for a complete set
of positive roots of $\frak{e}_{7(7)}$, an Iwasawa parametrization for the
representative of the irreducible, Riemannian, globally symmetric coset space%
\footnote{%
To be more precise, it is worth mentioning that the coset manifold
parametrized \textit{\`{a} la Iwasawa} by (\ref{coset}) is actually $%
\displaystyle{\frac{E_{7(7)}}{SU(8)/\mathbb{Z}_{2}}}$. The extra factor $%
\mathbb{Z}_{2}$ is due to the fact that our construction is performed by
starting from a completely non-compact Cartan subalgebra of $SL(8,\mathbb{R}%
) $, and thus the associated maximal torus is the double cover of that of $%
E_{7(7)}$. In other words, from the point of view of the corresponding
supergravity theory, spinors transform according to the double cover of the
stabilizer of the scalar manifold (see \textit{e.g.} \cite{Yo, AFZ}, and
Refs. therein).} (rank-$7$, dim$_{\mathbb{R}}=133$; see \textit{e.g.} \cite
{Helgason}, and also \cite{FM} for a recent review and list of Refs.)
\begin{equation}
\mathbf{C}\equiv \frac{E_{7(7)}}{SU(8)}  \label{def-coset}
\end{equation}
can be written down as follows\footnote{%
In the parametrization (\ref{coset}), the Abelianity of the completely
non-compact Cartan subalgebra $\mathcal{C}$ of $\frak{e}_{7\left( 7\right) }$
defined by (\ref{C-def}) has been implemented:
\begin{equation*}
\prod_{\alpha =1}^{7}\exp \left( x^{\alpha }h_{D_{\alpha }}\right) =\exp
\left( \sum_{\alpha =1}^{7}x^{\alpha }h_{D_{\alpha }}\right) .
\end{equation*}
} (as above, $1\leq \alpha \leq 7$, $1\leq i<j\leq 8$, and $I\in \mathcal{I}%
_{8}$):
\begin{equation}
\mathbf{C}(x^{\alpha },x^{ij},x^{I})=\exp \left( \sum_{\alpha
=1}^{7}x^{\alpha }h_{D_{\alpha }}\right) \prod_{i<j}\exp \left(
x^{ij}J_{ij}^{+}\right) \prod_{I\in \mathcal{I}_{8}}\exp \left( x^{I}%
\mathcal{J}_{I}^{+}\right) .  \label{coset}
\end{equation}

The explicit expression of the three typologies of matrices $h_{D_{\alpha }}$%
, $J_{ij}^{+}$ and $\mathcal{J}_{I}^{+}$ appearing in (\ref{coset}) can be
obtained as follows.

Let us start by observing that all such $56\times 56$ matrices can be
rewritten in terms of $28\times 28$ blocks $M_{ij}^{mn}$ (with $1\leq
m<n\leq 8$ selecting the columns, and $1\leq i<j\leq 8$ respectively the
rows). In particular
\begin{equation}
\delta _{ij}^{mn}\equiv \frac{1}{2}(\delta _{i}^{m}\delta _{j}^{n}-\delta
_{j}^{m}\delta _{i}^{n})  \label{antisymm-id-def}
\end{equation}
is the antisymmetric identity. Thus, the diagonal generators of completely
non-compact Cartan subalgebra $\mathcal{C}$ of $\frak{e}_{7\left( 7\right) }$
read (recall Eqs. (\ref{PPPA-3}) and (\ref{3}))
\begin{equation}
h_{D_{\alpha }}\equiv \left(
\begin{array}{cc}
(D_{\alpha }^{i}+D_{\alpha }^{j})\delta _{ij}^{mn} & 0 \\
0 & -(D_{\alpha }^{i}+D_{\alpha }^{j})\delta _{mn}^{ij}
\end{array}
\right) .  \label{D}
\end{equation}
On the other hand, the $J_{kl}^{+}$ matrices read ($1\leq k<l\leq 8$; recall
Eqs. (\ref{PPPA-1})-(\ref{PPPA-2}) and (\ref{1})-(\ref{2}), as well as
definitions (\ref{PPPA-4})-(\ref{PPPA-6}) and (\ref{PPP-eve-1}))
\begin{equation}
J_{kl}^{+}=\left(
\begin{array}{cc}
M_{ij}^{~~mn} & 0 \\
0 & N_{~~mn}^{ij}
\end{array}
\right) \equiv \sqrt{2}\left(
\begin{array}{cc}
\delta _{li}\delta _{kj}^{mn}-\delta _{lj}\delta _{ki}^{mn} & 0 \\
0 & \delta _{kn}\delta _{lm}^{ij}-\delta _{km}\delta _{ln}^{ij}
\end{array}
\right) ,  \label{PPPP-1}
\end{equation}
thus implying
\begin{equation}
(J_{kl}^{+})^{2}=0.
\end{equation}
Concerning the remaining matrices, definition (\ref{J+-def}) implies $%
\mathcal{J}_{I}^{+}\equiv \lambda _{I}$. Thus, by recalling (\ref{MI-1}) one
obtains
\begin{equation}
\mathcal{J}{}_{i_{1}i_{2}i_{3}i_{4}}^{+}=\left(
\begin{array}{cc}
0 & \lambda _{u}(i_{1},i_{2},i_{3},i_{4})_{ij\mid mn} \\
\lambda _{d}(i_{1},i_{2},i_{3},i_{4})^{ij\mid mn} & 0
\end{array}
\right) \equiv \left(
\begin{array}{cc}
0 & \epsilon _{i_{1}i_{2}i_{3}i_{4}ijmn} \\
\delta _{i_{1}i_{2}i_{3}i_{4}}^{ijmn} & 0
\end{array}
\right) ,  \label{MI-2}
\end{equation}
again implying
\begin{equation}
(\mathcal{J}{}_{i_{1}i_{2}i_{3}i_{4}}^{+})^{2}=0.
\end{equation}

Attention should be paid to the fact that, consistently with the very
definition (\ref{def-coset}) of coset representative $\mathbf{C}$ (\cite{dWN}%
; also cfr. Eq. (3.1) of \cite{CFGM1}), in all expressions (\ref
{antisymm-id-def}), (\ref{D}), (\ref{PPPP-1}) and (\ref{MI-2}) $ij$ are $%
SU(8)$-indices, whereas $mn$ are $E_{7\left( 7\right) }$-indices. However,
note that, within the explicit coset construction \textit{\`{a} la Iwasawa}
performed above, the indices $ij$ actually belong to a representation of $%
SL(8,\mathbb{R})$, so they are not fully $SU(8)$-covariant, but rather
covariant only under $SO(8)=SU(8)\cap SL(8,\mathbb{R})$. This covariance is
manifest by construction. The physical implications within the theory of
extremal black hole attractors in $\mathcal{N}=8$, $d=4$ supergravity will
be discussed in the next Sections.

Furthermore, it should be remarked that the difference among $E_{7\left(
7\right) }$-covariant and $SU(8)$-covariant indices can be removed by
suitably performing an $SU\left( 8\right) $-gauge-fixing and then retaining
only manifest invariance with respect to the rigid diagonal subgroup of $%
E_{7\left( 7\right) }\times SU\left( 8\right) $, without distinction among
the two types of indices (see \textit{e.g.} \cite{dWN}, as well as Sect. 3
of \cite{CFGM1}; see also \cite{KS-1}).

Let us also point out that all matrices constructed so far are
infinitesimally symplectic, \textit{i.e.} they belong to $\frak{sp}\left( 56,%
\mathbb{R}\right) $. In fact, their structure reads
\begin{equation}
M=\left(
\begin{array}{cc}
A & B \\
C & -A^{T}
\end{array}
\right) ,
\end{equation}
where $B$ and $C$ are $28\times 28$ symmetric matrices, and thus they do
satisfy the infinitesimal symplectic condition:
\begin{equation}
M^{T}\Omega +\Omega M=0,
\end{equation}
where $\Omega $ is the symplectic metric ($\mathbf{I}$ denoting the $%
28\times 28$ identity):
\begin{equation}
\Omega \equiv \left(
\begin{array}{cc}
0 & -\mathbf{I} \\
\mathbf{I} & 0
\end{array}
\right) .  \label{sympl-metric-def}
\end{equation}
In particular, the embedding of $\frak{su}(8)$ into $\frak{sp}\left( 56,%
\mathbb{R}\right) $ symplectic algebra is provided by the $56\times 56$
matrices of the form
\begin{equation}
U=\left(
\begin{array}{cc}
X & Y \\
-Y & X
\end{array}
\right) ,
\end{equation}
where $X$ is a linear combination of the $28$ antisymmetric $28\times 28$
matrices $A_{kl}$'s (defined by (\ref{A-antisymm-def}), and then extended by
(\ref{PPPA-1}) and (\ref{PPPA-4})-(\ref{PPPA-6})), while $Y$ is a linear
combination of the $28$ symmetric $28\times 28$ matrices $S_{kl}$ (defined
by (\ref{S-symm-def}), and then extended by (\ref{PPPA-2}) and (\ref{PPPA-4}%
)-(\ref{PPPA-6})) and of the $7$\ diagonal traceless matrices $D_{\alpha }$\
(defined by (\ref{D-Tr}), and then extended by (\ref{PPPA-3}) and (\ref{3})).


\section{\label{N=8,d=4-sugra}$\mathcal{N}=8$, $d=4$ Supergravity \textit{%
\`{a} la Iwasawa}}

Consistent with the notation of \cite{CFGM1}, one can rewrite the coset%
\footnote{%
An equivalent Iwasawa parametrization might be
\begin{equation*}
\mathbf{C}\left( x^{\alpha },x^{kl},x^{I}\right) =\exp \left( \sum_{\alpha
=1}^{7}x^{\alpha }h_{D_{\alpha }}\right) \exp \left(
\sum_{k<l=1}^{8}x^{kl}J_{kl}^{+}\right) \exp \left( \sum_{I\in \mathcal{I}%
_{8}}x^{I}\mathcal{J}_{I}^{+}\right) .
\end{equation*}
However, for computational purposes, it is more convenient to choose the
product of the exponentials of the generators, rather than the exponential
of their linear combination.} (\ref{coset}) as:
\begin{eqnarray}
\mathbf{C}\left( x^{\alpha },x^{ij},x^{I}\right) &=&\exp \left( \sum_{\alpha
=1}^{7}x^{\alpha }h_{D_{\alpha }}\right) \prod_{i<j}\exp \left(
x^{ij}J_{ij}^{+}\right) \prod_{I\in {}\mathcal{I}_{8}}\exp \left( x^{I}%
\mathcal{J}_{I}^{+}\right)  \notag \\
&&  \notag \\
&\equiv &\frac{1}{\sqrt{2}}\left(
\begin{array}{ccc}
\left( W_{1}\right) _{ij}^{~~mn} & ~ & \left( V_{1}\right) _{ij\mid mn} \\
~ & ~ & ~ \\
\left( V_{2}\right) ^{ij\mid mn} & ~ & \left( W_{2}\right) _{~~mn}^{ij}
\end{array}
\right) ,  \label{W}
\end{eqnarray}
where $i,j=1,...,8$ are in the fundamental irrepr. $\mathbf{8}$ of $SL\left(
8,\mathbb{R}\right) $, and all the indices are antisymmetrized, \textit{i.e.}
$ij=\left[ ij\right] $, $mn=\left[ mn\right] $ is understood throughout.

As pointed out in Eq. (2.1) of \cite{CFGM1}, the block-writing (\ref{W}) of
the $\frac{E_{7\left( 7\right) }}{SU\left( 8\right) }$-coset representative $%
\mathbf{C}$ corresponds to the following branching of the (fundamental
irrepr. $\mathbf{56}$ of the) $\mathcal{N}=8$, $d=4$ $U$-duality group:
\begin{equation}
\begin{array}{l}
E_{7\left( 7\right) }\supsetneq _{symm}^{\max }SL\left( 8,\mathbb{R}\right)
\\
\\
\mathbf{56}\rightarrow \mathbf{28}+\mathbf{28}^{\prime },
\end{array}
\end{equation}
where the prime denotes the contragradient irrepr(s.) throughout. Thus, in
the symplectic basis under consideration, the maximal symmetry of the
Lagrangian density of $\mathcal{N}=8$, $d=4$ supergravity is $SL\left( 8,%
\mathbb{R}\right) $, whose \textit{maximal compact subgroup }($mcs$)\textit{%
\ - }with symmetric embedding -\textit{\ }reads:
\begin{equation}
SO\left( 8\right) =mcs\left( SL\left( 8,\mathbb{R}\right) \right) =SU\left(
8\right) \cap SL\left( 8,\mathbb{R}\right) ,  \label{inters-1}
\end{equation}
as also noticed in \cite{CFGM1} (see Sects. 2 and 3 therein), and mentioned
above. The $SU\left( 8\right) $ symmetry is recovered only \textit{on-shell}%
, and it is clearly the maximal compact (local) symmetry of the non-linear
sigma model of scalars (whose representative is $\mathbf{C}$). This is the
very same situation as in the de Wit-Nicolai-'s framework (see \cite{dWN},
and also \cite{CFGM1} and \cite{CFGn-1} for recent treatments).

It is possible to switch to a complex manifestly $SU\left( 8\right) $%
-covariant basis through a \textit{Cayley transformation}\footnote{%
Notice that usually a Cayley rotation is represented by a matrix of the form
\begin{equation*}
\tilde{R}\equiv \frac{1}{\sqrt{2}}\left(
\begin{array}{ccc}
\mathbf{I} & ~ & i\mathbf{I} \\
~ & ~ & ~ \\
\mathbf{I} & ~ & -i\mathbf{I}
\end{array}
\right) .
\end{equation*}
However, the map (\ref{our-Cayley}) that we use is a transformation from the
Siegel upper half-plane to the unit disk, as well. Thus, for our purposes (%
\ref{our-Cayley}) is equivalent to a standard Cayley rotation, and it yields
the consistent expressions (\ref{fh-1})-(\ref{fh-2}) for the symplectic
sections $\mathbf{f}_{ij}^{~~mn}$ and $\mathbf{h}_{ij\mid mn}$}. Such a
rotation can be described by the unitary matrix:
\begin{equation}
R\equiv \frac{1}{\sqrt{2}}\left(
\begin{array}{ccc}
\mathbf{I} & ~ & i\mathbf{I} \\
~ & ~ & ~ \\
i\mathbf{I} & ~ & \mathbf{I}
\end{array}
\right) \Leftrightarrow R^{-1}\equiv \frac{1}{\sqrt{2}}\left(
\begin{array}{ccc}
\mathbf{I} & ~ & -i\mathbf{I} \\
~ & ~ & ~ \\
-i\mathbf{I} & ~ & \mathbf{I}
\end{array}
\right) =\overline{R}.  \label{our-Cayley}
\end{equation}
Therefore, the manifestly $SU\left( 8\right) $-covariant $\frac{E_{7\left(
7\right) }}{SU\left( 8\right) }$-coset representative is \texttt{(cfr. Eq.
(3.1) of \cite{CFGM1})}
\begin{eqnarray}
\mathcal{V} &\equiv &R\,\,\mathbf{C}\left( x^{\alpha },x^{ij},x^{I}\right)
R^{-1}=  \notag \\
&&  \notag \\
&=&\frac{1}{2\sqrt{2}}\left(
\begin{array}{ccc}
\left[ W_{1}+W_{2}+i\left( V_{2}-V_{1}\right) \right] _{ij}^{~~mn} & ~ &
\left[ V_{1}+V_{2}-i\left( W_{1}-W_{2}\right) \right] _{ij\mid mn} \\
~ & ~ & ~ \\
\left[ V_{1}+V_{2}+i\left( W_{1}-W_{2}\right) \right] ^{ij\mid mn} & ~ &
\left[ W_{1}+W_{2}+i\left( V_{1}-V_{2}\right) \right] _{~~mn}^{ij}
\end{array}
\right)  \notag \\
&&  \notag \\
&\equiv &\left(
\begin{array}{ccc}
u_{ij}^{~~mn} & ~ & v_{ij\mid mn} \\
~ & ~ & ~ \\
v^{ij\mid mn} & ~ & u_{~~mn}^{ij}
\end{array}
\right) .  \label{V-call}
\end{eqnarray}

Now, by following the same procedure as in Eq. (3.2) of \cite{CFGM1} and
exploiting the symplectic formalism for extended supergravities recently
reviewed in \cite{ADFT-rev-1} (see therein for further Refs.), the electric
and magnetic symplectic sections of $\mathcal{N}=8$, $d=4$ supergravity can
be defined respectively as:
\begin{eqnarray}
\mathbf{f}_{ij}^{~~mn} &\equiv &\frac{1}{\sqrt{2}}\left( u+v\right)
_{ij}^{~~mn}=\frac{1}{4}\left[ \left( W_{1}+W_{2}+V_{1}+V_{2}\right)
+i\left( -W_{1}+W_{2}-V_{1}+V_{2}\right) \right] _{ij}^{~~mn};  \notag \\
&&  \label{fh-1} \\
\mathbf{h}_{ij\mid mn} &\equiv &-\frac{i}{\sqrt{2}}\left( u-v\right)
_{ij\mid mn}=\frac{1}{4}\left[ \left( W_{1}-W_{2}-V_{1}+V_{2}\right)
+i\left( -W_{1}-W_{2}+V_{1}+V_{2}\right) \right] _{ij\mid mn}.  \notag \\
&&  \label{fh-2}
\end{eqnarray}
Then, it can easily be checked by direct computation that the usual
relations among symplectic sections hold, namely (\textit{e.g.} cfr. Eqs.
(3.10) of \cite{CFGM1}):
\begin{eqnarray}
i\left( \mathbf{f^{\dag }h}-\mathbf{h^{\dag }f}\right) &=&\mathbf{I}; \\
\mathbf{h}^{T}\mathbf{f}-\mathbf{f}^{T}\mathbf{h} &=&0,
\end{eqnarray}
which also directly follows from the finite symplecticity ($Sp\left( 56,%
\mathbb{R}\right) $) condition satisfied by $\mathbf{C}$ itself (recall
definition (\ref{sympl-metric-def})):
\begin{equation}
\mathbf{C}^{T}\Omega \mathbf{C}=\Omega .
\end{equation}

Consequently, by recalling the definition of the $8\times 8$ antisymmetric
central charge matrix $Z_{ij}$ (see \textit{e.g.} Eq. (3.14) of \cite{CFGM1}%
, as well as the treatment of \cite{FK-N=8} and \cite{ADFT-rev-1}), one can
compute:
\begin{eqnarray}
Z_{ij}\left( x^{\alpha },x^{kl},x^{I};q_{mn},p^{mn}\right) &\equiv &\mathbf{f%
}_{ij}^{~~mn}q_{mn}-\mathbf{h}_{ij\mid mn}p^{mn}=  \notag \\
&=&\frac{1}{4}\left[
\begin{array}{l}
\left( W_{1}+W_{2}+V_{1}+V_{2}\right) + \\
+i\left( -W_{1}+W_{2}-V_{1}+V_{2}\right)
\end{array}
\right] _{ij}^{~~mn}\left( x^{\alpha },x^{kl},x^{I}\right) q_{mn}+  \notag \\
&&-\frac{1}{4}\left[
\begin{array}{l}
\left( W_{1}-W_{2}-V_{1}+V_{2}\right) + \\
+i\left( -W_{1}-W_{2}+V_{1}+V_{2}\right) +
\end{array}
\right] _{ij\mid mn}\left( x^{\alpha },x^{kl},x^{I}\right) p^{mn}.  \notag \\
&&
\end{eqnarray}
It follows that the positive definite effective black hole potential \cite
{FGK} can be written as follows (\cite{FK-N=8,ADFT-rev-1}; recall definition
(\ref{V_BH-def}), and cfr. Eq. (3.17) of \cite{CFGM1}, as well):
\begin{eqnarray}
V_{BH} &\equiv &\frac{1}{2}Tr\left( ZZ^{\dag }\right) =\frac{1}{2}Z_{ij}%
\overline{Z}^{ij}=  \notag \\
&=&\frac{1}{2^{5}}\left[ \left( W_{1}+W_{2}+V_{1}+V_{2}\right) +i\left(
-W_{1}+W_{2}-V_{1}+V_{2}\right) \right] _{ij}^{~~mn}\cdot  \notag \\
&&\cdot \left[ \left( W_{1}+W_{2}+V_{1}+V_{2}\right) -i\left(
-W_{1}+W_{2}-V_{1}+V_{2}\right) \right] ^{~~ij\mid rs}q_{mn}q_{rs}+  \notag
\\
&&  \notag \\
&&-\frac{1}{2^{5}}\left[ \left( W_{1}+W_{2}+V_{1}+V_{2}\right) +i\left(
-W_{1}+W_{2}-V_{1}+V_{2}\right) \right] _{ij}^{~~mn}\cdot  \notag \\
&&\cdot \left[ \left( W_{1}-W_{2}-V_{1}+V_{2}\right) -i\left(
-W_{1}-W_{2}+V_{1}+V_{2}\right) +\right] _{~~rs}^{ij}q_{mn}p^{rs}+  \notag \\
&&  \notag \\
&&-\frac{1}{2^{5}}\left[ \left( W_{1}-W_{2}-V_{1}+V_{2}\right) +i\left(
-W_{1}-W_{2}+V_{1}+V_{2}\right) +\right] _{ij\mid rs}\cdot  \notag \\
&&\cdot \left[ \left( W_{1}+W_{2}+V_{1}+V_{2}\right) -i\left(
-W_{1}+W_{2}-V_{1}+V_{2}\right) \right] ^{~~ijmn}q_{mn}p^{rs}+  \notag \\
&&  \notag \\
&&+\frac{1}{2^{5}}\left[ \left( W_{1}-W_{2}-V_{1}+V_{2}\right) +i\left(
-W_{1}-W_{2}+V_{1}+V_{2}\right) +\right] _{ij\mid mn}\cdot  \notag \\
&&\cdot \left[ \left( W_{1}-W_{2}-V_{1}+V_{2}\right) -i\left(
-W_{1}-W_{2}+V_{1}+V_{2}\right) +\right] _{~~rs}^{ij}p^{mn}p^{rs}.  \notag \\
&&  \label{bhp}
\end{eqnarray}

\subsection{\label{Scalar-Covariance}Scalar Covariance}

Before proceeding further with the computations, a comment on the manifest
covariance properties of the $70$ real scalars $\left\{ x^{\alpha
},x^{kl},x^{I}\right\} $ in our parametrization \textit{\`{a} la Iwasawa}
given by $\mathbf{C}$ (\ref{W}) and $\mathcal{V}$ (\ref{V-call}) is needed.

It is worth recalling that in the de Wit-Nicolai parametrization (\cite{dWN}%
; see also Sects. 2 and 3 of \cite{CFGM1}, and Refs. therein), the $70$ real
scalars $\phi _{ijkl}=\phi _{\left[ ijkl\right] }$ coordinatizing $\frac{%
E_{7\left( 7\right) }}{SU\left( 8\right) }$ sit in the rank-$4$
antisymmetric irrepr. $\mathbf{70}$ of $SU\left( 8\right) $, which is
constrained by a self-reality condition (see \textit{e.g.} Eq. (3.4) of \cite
{FK-N=8}):
\begin{equation}
\overline{\phi }^{ijkl}\equiv \frac{1}{4!}\epsilon ^{ijklmnpq}\phi _{mnpq}.
\end{equation}

On the other hand, in the explicit construction and subsequent Iwasawa
parametrization performed above, the scalars have different types of
indices, and thus different covariance properties, namely:
\begin{equation}
\left\{ x^{\alpha },x^{kl},x^{I}\right\} ,  \label{scalars-1}
\end{equation}
where:

\begin{itemize}
\item  $I\in \mathcal{I}_{8}$ is in the rank-$3$ antisymmetric irreducible
repr. $\mathbf{35}$ of $SL\left( 7,\mathbb{R}\right) $;

\item  $kl=\left[ kl\right] $ is in the rank-$2$ antisymmetric
(contra-gradient) $\mathbf{28}^{\prime }$ irrepr. of $SL\left( 8,\mathbb{R}
\right) $;

\item  $\alpha $ is in the fundamental irrepr. $\mathbf{7}$ of $SL\left( 7,%
\mathbb{R}\right) $.
\end{itemize}

Thus, the maximal common covariance of the scalars (\ref{scalars-1}) is $%
SL\left( 7,\mathbb{R}\right) $, yielding to the following split of indices $%
kl$ :
\begin{equation}
\begin{array}{l}
SL\left( 8,\mathbb{R}\right) \rightarrow SL\left( 7,\mathbb{R}\right) \times
SO\left( 1,1\right) \\
\\
\mathbf{28}^{\prime }\rightarrow \mathbf{21}_{1}^{\prime }+\mathbf{7}%
_{-3}^{\prime },
\end{array}
\label{br-1}
\end{equation}
where $\mathbf{21}$ is the rank-$2$ antisymmetric (contra-gradient) irrepr.
of $SL\left( 7,\mathbb{R}\right) $, and the subscripts denote the weights
with respect to $SO\left( 1,1\right) $.



\section{\label{0=1/8-BPS-large}The Origin of $E_{7\left( 7\right) }/SU(8)$
as $\frac{1}{8}$-BPS Attractor}

{F}rom Eq. (3.15) of \cite{CFGM1}, it follows that at the $\frac{1}{8}$-BPS
attractor solution (recall $i=1,...,8$ throughout)
\begin{equation}
\phi _{ijkl}=0  \label{origin}
\end{equation}
the $\mathcal{N}=8$, $d=4$ central charge matrix $Z_{ij}$ simply reads (cfr.
Eq. (3.16) of \cite{CFGM1})
\begin{equation}
\left. Z_{ij}\right| _{\phi =0}=\frac{1}{\sqrt{2}}Q_{ij}\equiv \frac{1}{2}%
\left( q_{ij}+ip^{ij}\right) .  \label{Z}
\end{equation}
The origin (\ref{origin}) of the the scalar manifold $\frac{E_{7\left(
7\right) }}{SU\left( 8\right) }$, as a $\frac{1}{8}$-BPS attractor solution,
is supported by the skew-diagonal charge configuration \cite{FG,FK-N=8,CFGM1}
\begin{equation}
Q_{ij}=\left(
\begin{array}{cc}
\mathcal{Q} & 0_{1\times 3} \\
0_{3\times 1} & 0_{3\times 3}
\end{array}
\right) \otimes \epsilon ,  \label{Q1}
\end{equation}
where
\begin{equation}
\mathcal{Q}\equiv Q_{12}\equiv \frac{1}{\sqrt{2}}\left(
q_{12}+ip^{12}\right) \equiv \frac{1}{\sqrt{2}}\left( q+ip\right) \in
\mathbb{C},  \label{Q2}
\end{equation}
and $\epsilon $ is the $2\times 2$ symplectic metric:
\begin{equation}
\epsilon \equiv \left(
\begin{array}{cc}
0 & -1 \\
1 & 0
\end{array}
\right) .  \label{epsilon-def}
\end{equation}
Thus, Eq. (\ref{Z}) implies that at the $\frac{1}{8}$-BPS attractor solution
(\ref{origin}) supported by the charge configuration (\ref{Q1})-(\ref{Q2}), $%
Z_{ij}$ reads as follows:
\begin{equation}
Z_{ij,\frac{1}{8}-BPS,\text{large}}=\left(
\begin{array}{cccc}
\frac{1}{\sqrt{2}}\left( q+ip\right) &  &  &  \\
& 0 &  &  \\
&  & 0 &  \\
&  &  & 0
\end{array}
\right) \otimes \epsilon .  \label{Z-BPS-large}
\end{equation}
By comparing Eq. (\ref{Z-BPS-large}) with the usual parametrization of $%
Z_{ij,\frac{1}{8}-BPS,\text{large}}$ (\label%
{Ferrara-Maldacena,FG,FK-N=8,BMP-1}; see \textit{e.g.} Eq. (2.20) of \cite
{FMO-1}), it follows that the absolute value and the phase of the unique
non-vanishing skew-eigenvalue of $Z_{ij,\frac{1}{8}-BPS,\text{large}}$ reads
\begin{eqnarray}
\rho _{BPS} &=&\frac{1}{\sqrt{2}}\left| q+ip\right| =\frac{\sqrt{q^{2}+p^{2}}%
}{\sqrt{2}};  \label{rho-BPS} \\
\frac{\varphi }{4} &=&\psi _{Q_{12}}:\tan \left( \psi _{Q_{12}}\right) =%
\frac{p}{q}.  \label{n-1}
\end{eqnarray}

Concerning (\ref{n-1}), the stabilization of the phase $\varphi $ in
terms of the ratio of the charges $\frac{p}{q}$ might seem to be
inconsistent with the well known fact that the overall phase of
$\mathcal{N}=8$, $d=4$ central charge matrix $Z_{ij}$ is
undetermined at $\frac{1}{8}$-BPS attractors \cite
{Ferrara-Maldacena,FG,FK-N=8}. But it should be recalled that when
taking the Lie algebra limit, the scalar configuration (\ref{origin}) is
picked as the starting point, which fixes the $\frac{1}{8}$-BPS attractor to be the origin of the scalar manifold $M_{%
\mathcal{N}=8,d=4}=\frac{E_{7\left( 7\right) }}{SU\left( 8\right) }$
(recall Eq. (\ref{scalar-manifold})). As it is well known,
\textit{not all} scalar fields are stabilized in terms of the
electric and magnetic charges at the event horizon of
$\frac{1}{8}$-BPS attractors, but rather a certain subset of them
spans a related \textit{moduli space} (see also Sect.
\ref{d=4-d=5}).
According to \cite{ADF-U-duality-d=4,Ferrara-Marrani-2,CFGM1,ICL-1}, the \textit{%
moduli space} of $\frac{1}{8}$-BPS attractors $\mathcal{M}_{\frac{1}{8}-BPS,%
\text{large}}$ is a symmetric quaternionic space of real dimension
$40$, given by Eqs. (\ref{pre-new-rel-2})-(\ref{new-rel-2}) below.
In the maximal symmetric embedding
\cite{ADF-U-duality-d=4,Ferrara-Marrani-1,CFGM1}
\begin{equation}
\begin{array}{l}
SU\left( 8\right) \supsetneq SU\left( 6\right) \times SU\left(
2\right)
\times U\left( 1\right)  \\
\\
\mathbf{70}\rightarrow \left( \mathbf{15},\mathbf{1}\right)
_{-2}+\left(
\overline{\mathbf{15}},\mathbf{1}\right) _{2}+\left( \mathbf{20},\mathbf{2}%
\right) _{0},
\end{array}
\end{equation}
the $40$ real scalar degrees of freedom pertaining to $\mathcal{M}_{\frac{1}{%
8}-BPS,\text{large}}$ fit into the $\left(
\mathbf{20},\mathbf{2}\right) _{0} $ irrepr. of $SU\left( 6\right)
\times SU\left( 2\right) \times U\left( 1\right) $, where
$\mathbf{15}$, $\overline{\mathbf{15}}$ and $\mathbf{20}$
respectively are the rank-$2$ antisymmetric, complex conjugate
rank-$2$ antisymmetric, and rank-$3$ antisymmetric irreprs. of
$SU\left( 6\right) $, and the subscripts denote the $U\left(
1\right) $ charges. Thus, by specifying (\ref{origin}), both the
$30$ stabilized scalar degrees of
freedom (coordinatizing the complementary to $\mathcal{M}_{\frac{1}{8}-BPS,%
\text{large}}$ in $M_{\mathcal{N}=8,d=4}$) and the $40$
\textit{would-be unstabilized} scalar degrees of freedom (coordinatizing the submanifold $%
\mathcal{M}_{\frac{1}{8}-BPS,\text{large}}\subsetneq
M_{\mathcal{N}=8,d=4}$)
are set to zero. Such a ``polarization'' of the \textit{moduli space} $%
\mathcal{M}_{\frac{1}{8}-BPS,\text{large}}$ implies the phase
$\varphi $ to be actually \textit{determined} in terms of relevant
charges, as given by Eq. (\ref{n-1}).

Furthermore, the maximal compact symmetry exhibited by $Z_{ij,\frac{1}{8}%
-BPS,\text{large}}$ given by Eq. (\ref{Z-BPS-large}) (or equivalently,
through Eq. (\ref{Z}), by $Q_{ij}$ given by Eqs. (\ref{Q1}) and (\ref{Q2}))
is
\begin{equation}
SU\left( 6\right) \times SU\left( 2\right) =mcs\left( E_{6\left( 2\right)
}\right) ,  \label{1/8-BPS-large-mcs}
\end{equation}
(see \textit{e.g.} \cite{Gilmore}). In other words, $Q_{ij}$ given by Eqs. (%
\ref{Q1}) and (\ref{Q2})) is a representative of the $\frac{1}{8}$-BPS
``large'' (\textit{i.e.} attractive) orbit of the $\mathbf{56}$ fundamental
representation space of $E_{7\left( 7\right) }$, exhibiting the maximal
(compact) symmetry (\ref{1/8-BPS-large-mcs}) \cite{FG,LPS-1}:
\begin{equation}
Q_{ij}\in \mathcal{O}_{\frac{1}{8}-BPS,\text{large}}=\frac{E_{7\left(
-7\right) }}{E_{6\left( 2\right) }}.  \label{Q-1/8-BPS-large}
\end{equation}

Notice that $\mathcal{Q}$ and $\overline{\mathcal{Q}}$ are charged with
respect to the $U\left( 1\right) $ in the extreme right-hand side of the
following chain of maximal symmetric group embeddings (see \textit{e.g.}
\cite{Gilmore})
\begin{equation}
E_{7\left( 7\right) }\overset{mcs}{\supsetneq }SU\left( 8\right) \supsetneq
SU\left( 6\right) \times SU\left( 2\right) \times U\left( 1\right) .
\label{emb-1a}
\end{equation}
Indeed, $\mathcal{Q}$ and $\overline{\mathcal{Q}}$ respectively are the $%
\left( SU\left( 6\right) \times SU\left( 2\right) \right) $-singlets $\left(
\mathbf{1},\mathbf{1}\right) _{-3}$ and $\left( \mathbf{1},\mathbf{1}\right)
_{+3}$ in the subsequent decomposition of the fundamental irrepr. $\mathbf{56%
}$ of $E_{7\left( 7\right) }$ \cite{CFGM1}:
\begin{eqnarray}
E_{7\left( 7\right) } &\rightarrow &SU\left( 8\right) \rightarrow SU\left(
6\right) \times SU\left( 2\right) \times U\left( 1\right) ;  \notag \\
&&  \notag \\
\mathbf{56} &\rightarrow &\mathbf{28}+\overline{\mathbf{28}}\rightarrow
\left( \mathbf{15},\mathbf{1}\right) _{+1}+\left( \mathbf{6},\mathbf{2}%
\right) _{-1}+\left( \mathbf{1},\mathbf{1}\right) _{-3}+  \notag \\
&&+\left( \overline{\mathbf{15}},\mathbf{1}\right) _{-1}+\left( \overline{%
\mathbf{6}},\mathbf{2}\right) _{+1}+\left( \mathbf{1},\mathbf{1}\right)
_{+3},  \label{BPS-decomp}
\end{eqnarray}
where the subscripts denote the charge with respect to $U\left( 1\right) $.
Thus, by suitably renormalizing the $U\left( 1\right) $-phase $\Phi $, $%
\mathcal{Q}$ and\ $\overline{\mathcal{Q}}$ transform under considered $%
U\left( 1\right) $ respectively as follows:
\begin{equation}
U\left( 1\right) :\left\{
\begin{array}{l}
\mathcal{Q}\longrightarrow \mathcal{Q}e^{-i\Phi }; \\
\\
\overline{\mathcal{Q}}\longrightarrow \overline{\mathcal{Q}}e^{i\Phi },
\end{array}
\right.  \label{U(1)}
\end{equation}
or, equivalently, in the $\mathbf{2}$ of $SO\left( 2\right) $:
\begin{equation}
SO\left( 2\right) :\left(
\begin{array}{c}
q \\
p
\end{array}
\right) \longrightarrow \left(
\begin{array}{cc}
\cos \Phi & \sin \Phi \\
-\sin \Phi & \cos \Phi
\end{array}
\right) \left(
\begin{array}{c}
q \\
p
\end{array}
\right) .  \label{SO(2)}
\end{equation}
Thus, the $U\left( 1\right) $-transformation properties of $\rho _{BPS}$ and
$\psi _{Q_{12}}$, given by Eqs. (\ref{rho-BPS}) and (\ref{n-1}),
respectively are
\begin{eqnarray}
\rho _{BPS} &=&\sqrt{\mathcal{Q}\overline{\mathcal{Q}}}\longrightarrow \rho
_{BPS}; \\
\psi _{Q_{12}} &\longrightarrow &\arctan \left( \frac{p}{q}\right) -\Phi .
\label{PSI-U(1)}
\end{eqnarray}
As a consequence, $\rho _{BPS}$ is invariant under $SU\left( 6\right) \times
SU\left( 2\right) \times U(1)$, whereas $\psi _{Q_{12}}$ only under $%
SU\left( 6\right) \times SU\left( 2\right) $. Notice that the fact that $%
\psi _{Q_{12}}$ is not invariant under the considered $U\left(
1\right) $ (as given by Eq. (\ref{PSI-U(1)})) does not affect the
attractor nature of the origin (\ref{origin}) of $\frac{E_{7\left(
7\right) }}{SU\left( 8\right) }$, because this point is $SU\left(
8\right)$-invariant (see Footnote 12). As mentioned in the
Introduction (see also final Sect. \ref{Conclusion}), the $U\left(
1\right) $ under consideration should be seen as a residual
``degeneracy symmetry'' at $\frac{1}{8}$-BPS attractor solutions in
the Lie algebra limit (namely, at the origin (\ref{origin}) of the scalar
manifold $\frac{E_{7\left( 7\right) }}{SU\left( 8\right) }$).


\subsection{\label{Iwa}Iwasawa Solutions}

We are now going to analyze the consequences of the above reasoning for the
explicit construction \textit{\`{a} la Adams-Iwasawa} performed in Subsect.
\ref{Lie-Algebra}.

Thus, we start and study how the $1/8$-BPS Attractor Eqs. can be implemented
in the Iwasawa construction performed above, confining ourselves to the
origin (\ref{origin}) of $\frac{E_{7\left( 7\right) }}{SU\left( 8\right) }$
as attractor solution. {Therefore, we have to set all scalar fields to zero
in the parametrization (\ref{W}) of the coset representative }$\mathbf{C}$.
Since the Attractor Eqs. are nothing but criticality conditions for the
effective black hole potential $V_{BH}$ \cite{FGK}, namely
\begin{equation}
\partial _{\phi }V_{BH}=0:\left. V_{BH}\right| _{\partial _{\phi
}V_{BH}=0}\neq 0,  \label{AEs}
\end{equation}
this implies that one only needs to compute the terms of $V_{BH}$ which are
\textit{linear} in the scalar fields. To the first order in scalar fields,
the coset representative $\mathbf{C}$ (\ref{W}) reads ($\mathbf{I}_{56}$
denoting the $56\times 56$ identity)
\begin{equation}
\mathbf{C}\left( x^{\alpha },x^{ij},x^{I}\right) =\mathbf{I}%
_{56}+\sum_{\alpha =1}^{7}x^{\alpha }h_{D_{\alpha
}}+\sum_{i<j=2}^{8}x^{ij}J_{ij}^{+}+\sum_{I\in \mathcal{I}_{8}}x^{I}\mathcal{%
J}_{I}^{+}+\mathcal{O}\left( \left\{ x^{\alpha },x^{ij},x^{I}\right\}
^{2}\right) .  \label{C-first-order}
\end{equation}
Thus, by neglecting $\mathcal{O}\left( \left\{ x^{\alpha
},x^{ij},x^{I}\right\} ^{2}\right) $, Eq. (\ref{W}) yields that
\begin{eqnarray}
(W_{1})_{ij}^{\ \ mn} &=&\mathbf{I}+\sum_{\alpha =1}^{7}x^{\alpha
}(D_{\alpha }^{i}+D_{\alpha }^{j})\delta
_{ij}^{mn}+\sum_{k<l=2}^{8}x^{kl}(J_{kl}^{+})_{ij}^{~~mn};  \label{expr-1} \\
(W_{2})_{\ \ mn}^{ij} &=&\mathbf{I}-\sum_{\alpha =1}^{7}x^{\alpha
}(D_{\alpha }^{i}+D_{\alpha }^{j})\delta
_{mn}^{ij}-\sum_{k<l=2}^{8}x^{kl}(J_{kl}^{+T})_{~~mn}^{ij};  \label{expr-2}
\\
\left( V_{1}\right) _{ij\mid mn} &=&\sum_{I\in \mathcal{I}_{8}}x^{I}\epsilon
_{Iijmn};  \label{expr-3} \\
\left( V_{2}\right) ^{ij\mid mn} &=&\sum_{I\in \mathcal{I}_{8}}x^{I}\epsilon
_{I}^{ijmn}.  \label{expr-4}
\end{eqnarray}
By plugging results (\ref{expr-1})-(\ref{expr-4}) into Eq. (\ref{bhp}) and
evaluating the resulting Attractor Eqs. (\ref{AEs}), the following system is
obtained:
\begin{equation}
\left\{
\begin{array}{l}
(h_{D_{\alpha }})_{\ \ rs}^{mn}q_{mn}p^{rs}=0; \\
\\
(J_{kl}^{+}+J_{kl}^{+T})_{\ \ rs}^{mn}q_{mn}p^{rs}=0; \\
\\
(\epsilon +\delta )_{I}^{\ mnrs}q_{mn}q_{rs}-(\epsilon +\delta
)_{Imnrs}p^{mn}p^{rs}=0.
\end{array}
\right.  \label{Iwa-AEs}
\end{equation}
By introducing complex charges $Q^{ij}$ (recall Eqs. (\ref{Q1}), (\ref{Q2}))
such that
\begin{equation}
q_{ij}=\frac{1}{\sqrt{2}}(Q^{ij}+\bar{Q}^{ij});\qquad p_{ij}=\frac{1}{i\sqrt{%
2}}(Q^{ij}-\bar{Q}^{ij}),
\end{equation}
the system (\ref{Iwa-AEs}) can be recast in the following form:
\begin{equation}
\left\{
\begin{array}{l}
I.~~(h_{D_{\alpha }})_{mn|rs}(Q^{mn}Q^{rs}-\bar{Q}^{mn}\bar{Q}^{rs})=0; \\
\\
II.~~(J_{kl}^{+}+J_{kl}^{+T})_{mn|rs}(Q^{mn}Q^{rs}-\bar{Q}^{mn}\bar{Q}%
^{rs})=0; \\
\\
III.~~(\epsilon +\delta )_{Imnrs}(Q^{mn}Q^{rs}+\bar{Q}^{mn}\bar{Q}^{rs})=0,
\end{array}
\right.  \label{Iwa-AEs-2}
\end{equation}
where the symmetry of the matrices involved was used. The system (\ref
{Iwa-AEs-2}) is nothing but the set of (necessarily $\frac{1}{8}$-BPS)
Attractor Eqs. consistent with the origin (\ref{origin}) of $\frac{%
E_{7\left( 7\right) }}{SU\left( 8\right) }$. Thus, Eqs. (\ref{Iwa-AEs-2})
express conditions on the complex dyonic charges $Q^{ij}$ such that
\begin{equation}
\phi _{ijkl}=0\Leftrightarrow \left\{
\begin{array}{l}
x^{\alpha }=0; \\
x^{ij}=0; \\
x^{I}=0
\end{array}
\right.  \label{origin-2}
\end{equation}
is a (necessarily $\frac{1}{8}$-BPS) attractor scalar configuration in the
background of a static, spherically symmetric, asymptotically flat extremal
black hole of $\mathcal{N}=8$, $d=4$ (\textit{ungauged}) supergravity.

Consistent with Eq. (\ref{Q1}), we look for solutions of (\ref{Iwa-AEs-2})
within the following structural \textit{Ansatz}:
\begin{equation}
Q=e^{i\varphi /4}\text{diag}\left( r,0,0,0\right) \otimes \epsilon ,
\label{sol-Ansatz}
\end{equation}
with $r\in \mathbb{R}_{0}^{+}$ and $\varphi \in \left[ 0,8\pi \right) $. By
recalling Eq. (\ref{Q-1/8-BPS-large}), $Q$ is a singlet of $SU(2)\times
SU(6) $, stabilizer of the $\frac{1}{8}$-BPS ``large'' orbit $\mathcal{O}_{%
\frac{1}{8}-BPS,\text{large}}$ \cite{FG,LPS-1}. By inserting (\ref
{sol-Ansatz}) into the system (\ref{Iwa-AEs-2}), Eqs. $II$\ and $III$\ are
automatically satisfied, whereas Eqs. $I$ yield the conditions
\begin{equation}
(D_{\alpha }^{1}+D_{\alpha }^{2})\sin \frac{\varphi }{2}=0,~~\forall \alpha
=1,...,7.  \label{sol-cond}
\end{equation}
Note that $D_{\alpha }^{1}+D_{\alpha }^{2}$ cannot vanish for all $\alpha $%
's, because $D_{\alpha }$'s defined by (\ref{D-Tr}) are a basis for the $%
8\times 8$ diagonal traceless matrices with real entries: if $D_{\alpha
}^{1}+D_{\alpha }^{2}=0$ $\forall \alpha =1,...,7$, then the $D_{\alpha }$'s
would not generate, \textit{e.g.}, the diagonal traceless matrix $h=$diag$%
\{1,1,-2,0,0,0,0,0\}$.

Thus, within the \textit{Ansatz} (\ref{sol-Ansatz}) exhibiting the maximal
(compact) symmetry $SU(2)\times SU(6)$, the solution of conditions (\ref
{sol-cond}), and thus of Attractor Eqs. (\ref{Iwa-AEs-2}), reads
\begin{equation}
\left\{
\begin{array}{l}
r\in \mathbb{R}_{0}; \\
\\
\varphi =2n\pi ,~n\in \mathbb{Z.}
\end{array}
\right.  \label{sol-AEs}
\end{equation}
Consequently, two kinds of solutions are obtained:

\begin{enumerate}
\item  the \textit{purely electric} solution:
\begin{equation}
\left\{
\begin{array}{l}
r=q; \\
\\
\varphi =4k\pi ,~k\in \mathbb{Z.}
\end{array}
\right.  \label{el-sol}
\end{equation}

\item  the \textit{purely magnetic} solution:
\begin{equation}
\left\{
\begin{array}{l}
r=p; \\
\\
\varphi =2\left( 2k+1\right) \pi ,~k\in \mathbb{Z.}
\end{array}
\right.  \label{magn-sol}
\end{equation}
\end{enumerate}

\subsection{\label{Breaking}Analysis of Solutions and Breaking $U\left(
1\right) \longrightarrow \mathbb{Z}_{4}$}

In order to analyze the obtained solutions (\ref{el-sol}) and (\ref{magn-sol}%
), which support the origin (\ref{origin-2}) as an $\frac{1}{8}$-BPS
attractor, let us consider the following charge configuration:
\begin{equation}
\widetilde{Z}_{ij,\frac{1}{8}-BPS,\text{large}}=\left(
\begin{array}{cccc}
\frac{1}{\sqrt{2}}q &  &  &  \\
& 0 &  &  \\
&  & 0 &  \\
&  &  & 0
\end{array}
\right) \otimes \epsilon .  \label{Z-part}
\end{equation}
By recalling Eqs. (\ref{Z-BPS-large}), (\ref{rho-BPS}), (\ref{n-1}), this is
an$\frac{1}{8}$-BPS attractor solution at the origin (\ref{origin}) of $%
\frac{E_{7\left( 7\right) }}{SU\left( 8\right) }$, with
\begin{eqnarray}
\widetilde{\rho }_{BPS} &=&\frac{\left| q\right| }{\sqrt{2}}; \\
\frac{\widetilde{\varphi }}{4} &=&\widetilde{\psi }_{Q_{12}}=\left\{
\begin{array}{l}
0,~q>0; \\
\\
\pi ,~q<0.
\end{array}
\right.
\end{eqnarray}
\bigskip Thus, it is of purely electric nature, namely of the kind (\ref
{el-sol}) obtained above.

It should be stressed that (\ref{Z-part}) does \textit{not} exhibit the
symmetry enhancement
\begin{equation}
SU\left( 6\right) \times SU\left( 2\right) \longrightarrow SU\left( 6\right)
\times SU\left( 2\right) \times U\left( 1\right) ,
\label{no-symm-enhancem-1}
\end{equation}
despite the fact that, consistent with Eqs. (\ref{Q1})-(\ref{Q2}), it is
supported by
\begin{eqnarray}
\widetilde{Q}_{ij} &=&\left(
\begin{array}{cc}
\widetilde{\mathcal{Q}} & 0 \\
0 & 0_{3\times 3}
\end{array}
\right) \otimes \epsilon ; \\
\widetilde{\mathcal{Q}} &\equiv &\frac{1}{\sqrt{2}}q\in \mathbb{R}.
\end{eqnarray}
Indeed, under the commuting $U\left( 1\right) $ in the r.h.s. of (\ref
{no-symm-enhancem-1}), the explicit expression of $\widetilde{\rho }_{BPS}$
changes as follows:
\begin{equation}
\widetilde{\rho }_{BPS}\longrightarrow \frac{\left| q\cos \Phi -iq\sin \Phi
\right| }{\sqrt{2}},  \label{P-1}
\end{equation}
although its value remains the same because
\begin{equation}
\left| q\cos \Phi -iq\sin \Phi \right| =\left| q\right| .
\end{equation}
Eq. (\ref{P-1}) is a trivial consequence of the transformation described by
Eq. (\ref{SO(2)}):
\begin{equation}
SO\left( 2\right) :\left(
\begin{array}{c}
q \\
0
\end{array}
\right) \longrightarrow \left(
\begin{array}{cc}
\cos \Phi & \sin \Phi \\
-\sin \Phi & \cos \Phi
\end{array}
\right) \left(
\begin{array}{c}
q \\
0
\end{array}
\right) =\left(
\begin{array}{c}
q\cos \Phi \\
-q\sin \Phi
\end{array}
\right) \equiv \left(
\begin{array}{c}
q^{\prime } \\
p^{\prime }
\end{array}
\right) .
\end{equation}
On the other hand, $\widetilde{\psi }_{Q_{12}}$ transforms as follows
(recall Eq. (\ref{PSI-U(1)})):
\begin{equation}
\widetilde{\psi }_{Q_{12}}\longrightarrow \widetilde{\psi }_{Q_{12}}^{\prime
}\equiv \arctan \left( \frac{p^{\prime }}{q^{\prime }}\right) =-\Phi ~\left(
+k\pi ,~k\in \mathbb{Z}\right) .  \label{PSI-U(1)-ours}
\end{equation}
\medskip

Through this reasoning, we now address the following question: \textit{why
are the Iwasawa }$\frac{1}{8}$\textit{-BPS attractor solutions (\ref{el-sol}%
) and (\ref{magn-sol}) respectively purely electric and purely magnetic, and
thus non-dyonic?}

As we will see below, the non-dyonicity of Iwasawa solutions (\ref{el-sol})
and (\ref{magn-sol}) is intrinsic to the Adams-Iwasawa construction approach
to the $\frac{E_{7\left( 7\right) }}{SU\left( 8\right) }$-coset
representative, exploited in\ Sect. \ref{Adams-Iwasawa}. Indeed, such a
construction explicitly breaks the residual $U\left( 1\right) $-symmetry
(parametrizing the dyonic nature of $\frac{1}{8}$-BPS attractor solutions at
the coset origin (\ref{origin-2}); recall Eqs. (\ref{U(1)})-(\ref{PSI-U(1)}%
)) into a discrete subgroup $\mathbb{Z}_{4}$. Thus, as an intrinsic feature
of the Iwasawa parametrization we performed, the residual ``degeneracy
symmetry'' $U\left( 1\right) $ of the charge orbit $\mathcal{O}_{\frac{1}{8}%
-BPS,\text{large}}$ (\ref{Q-1/8-BPS-large}) supporting the origin (\ref
{origin-2}) of the coset $\frac{E_{7\left( 7\right) }}{SU\left( 8\right) }$
as a $\frac{1}{8}$-BPS attractor gets broken, and discretized, to a finite
subgroup $\mathbb{Z}_{4}$.

In order to show this, let us consider Eqs. (\ref{U(1)})-(\ref{PSI-U(1)}).
These Eqs. express the action of the $U\left( 1\right) $ symmetry, given by
the extra commuting $U\left( 1\right) $-factor in the second group embedding
of (\ref{BPS-decomp}):
\begin{equation}
SU(2)\times SU(6)\times U(1)_{\mathcal{A}}\subsetneq _{symm}^{\max }SU(8).
\label{U(1)_A}
\end{equation}
which is maximal (``$\max $'') and symmetric (``$symm$''). We denote such an
extra commuting $U\left( 1\right) $-factor with a subscript ``$\mathcal{A}$%
'', in order to discriminate it from other $U\left( 1\right) $'s we will
consider below. (\ref{U(1)_A}) means that the maximal compact symmetry of $%
\mathcal{O}_{\frac{1}{8}-BPS,\text{large}}$ (\ref{Q-1/8-BPS-large}) embeds
into the maximal compact (local) symmetry of the non-linear sigma model of
scalars (and maximal compact \textit{on-shell} symmetry of the whole $%
\mathcal{N}=8$, $d=4$ \textit{ungauged} Lagrangian density, as well; see the
end of Subsect. \ref{Iwasawa}) through an extra commuting $U\left( 1\right)
_{\mathcal{A}}$. Such a $U\left( 1\right) _{\mathcal{A}}$ is a residual
\textit{``degeneracy symmetry''}, in the sense that it spans all possible
charge configurations supporting $\frac{1}{8}$-BPS attractor points in $%
\frac{E_{7\left( 7\right) }}{SU\left( 8\right) }$ (these latter are
exemplified by the manifestly $SU\left( 8\right) $-symmetric\footnote{%
More precisely, the origin $\phi _{ijkl}=0$ (or equivalently (\ref{origin-2}%
)) is the \textit{unique} $SU\left( 8\right) $-symmetric point of $\frac{%
E_{7\left( 7\right) }}{SU\left( 8\right) }$, such that its symmetry $%
SU\left( 8\right) $ coincides with the stabilizer of the coset itself. In
other words, consistent with the geometric construction of the non-compact,
irreducible, Riemannian, globally symmetric coset $\frac{E_{7\left( 7\right)
}}{SU\left( 8\right) }$, the action of the stabilizer $SU\left( 8\right) $
has a \textit{unique} fixed point, \textit{i.e.} the origin $\phi _{ijkl}=0$
of the coset itself.} point of the scalar manifold (\ref{scalar-manifold}),
namely by the origin (\ref{origin-2})).

However, as discussed at the start of Sect. \ref{N=8,d=4-sugra}, as well as
in Subsect. \ref{Scalar-Covariance} (recall Eq. (\ref{br-1})), the performed
Adams-Iwasawa construction of the $\frac{E_{7\left( 7\right) }}{SU\left(
8\right) }$-coset representative $\mathbf{C}$ (\ref{W}) explicitly breaks
the maximal covariance from $SL(8,\mathbb{R})$ down to $SL(7,\mathbb{R})$.
Through the Cayley-like transformation (\ref{our-Cayley}), yielding to the
manifestly $SU\left( 8\right) $-covariant $\frac{E_{7\left( 7\right) }}{%
SU\left( 8\right) }$-coset representative $\mathcal{V}$ (\ref{V-call}), the
breaking (\ref{br-1}) becomes
\begin{equation}
SU(7)\times U(1)_{\mathcal{E}}\subsetneq _{symm}^{\max }SU(8).
\label{br-1-c}
\end{equation}
The $\frac{1}{8}$-BPS ``large'' charge orbit $\mathcal{O}_{\frac{1}{8}-BPS,%
\text{large}}$ (\ref{Q-1/8-BPS-large}) gives rise to a further decomposition
into $SU(6)$, and thus the embedding (\ref{br-1}) can be completed as
follows:
\begin{equation}
SU(6)\times U(1)_{\mathcal{B}}\times U(1)_{\mathcal{E}}\subsetneq
_{symm}^{\max }SU(7)\times U(1)_{\mathcal{E}}\subsetneq _{symm}^{\max }SU(8).
\label{br-1-c-complete}
\end{equation}

The symmetry breaking (\ref{br-1}), intrinsic to the performed Iwasawa
parametrization (exploited by singling the Cartan subalgebra $\mathcal{C}$ (%
\ref{C-def}) out), is crucial. Indeed, by recalling Eqs. (\ref{Q1})-(\ref{Q2}%
), the embeddings (\ref{U(1)_A}) and (\ref{br-1-c-complete}) necessarily
lead to conclude that (\textit{if any}) the residual symmetry group $\Gamma $
of $\frac{1}{8}$-BPS attractors in the Adams-Iwasawa construction performed
above must be given by the intersection:
\begin{equation}
\Gamma \equiv U(1)_{\mathcal{A}}\cap U(1)_{\mathcal{E}}.  \label{Gamma}
\end{equation}

We now show that $\Gamma $ actually is a $\mathbb{Z}_{4}$ discrete finite
group.

In order to do this, it should be remarked that the original $SL(8,\mathbb{R}%
)$ group on which Adams' construction is based (see Subsect. \ref
{Lie-Algebra}), is broken, by the considered Iwasawa parametrization (see
Subsect. \ref{Iwasawa}), down to $SL(7,\mathbb{R})\times SO(1,1)$ (as given
by (\ref{br-1})), where the generator of $SO\left( 1,1\right) $ can be
written as ($\mathbf{I}_{2}$, $\mathbf{I}_{6}$ and $\mathbf{I}_{7}$
respectively denoting the $2\times 2$, $6\times 6$ and $7\times 7$ identity)
\begin{equation}
\tau \equiv \left(
\begin{array}{cc}
\mathbf{I}_{7} & 0 \\
0 & -7
\end{array}
\right) .  \label{SO(1,1)}
\end{equation}
After a suitable Cayley transformation\footnote{%
To be more precise, $\mathbf{R}\equiv R_{1-4}R_{2-5}\mathbf{R}_{0}$, where $%
R_{i-j}$ is the rotation of $\pi /2$ of the $i-j$ plane and $\mathbf{R}_{0}$
is the $8\times 8$ Cayley transformation (also recall Footnote 11)
\begin{equation*}
\mathbf{R}_{0}\equiv \frac{1}{\sqrt{2}}\left(
\begin{array}{cc}
\mathbf{I}_{4} & -i\mathbf{I}_{4} \\
-i\mathbf{I}_{4} & \mathbf{I}_{4}
\end{array}
\right)
\end{equation*}
} $\mathbf{R}$, it takes the form:
\begin{eqnarray}
\tilde{\tau} &\equiv &\mathbf{R}\tau \mathbf{R}^{-1}=\left(
\begin{array}{cc}
\tau _{\mathcal{E}} & 0 \\
0 & \mathbf{I}_{6}
\end{array}
\right) ;  \label{our-Cayley-2} \\
\mathbf{R} &\equiv &\frac{1}{\sqrt{2}}\left(
\begin{array}{cccccccc}
0 & 0 & 0 & 1 & 0 & 0 & 0 & -i \\
0 & 0 & 0 & -i & 0 & 0 & 0 & 1 \\
0 & 0 & 1 & 0 & 0 & 0 & -i & 0 \\
-1 & 0 & 0 & 0 & i & 0 & 0 & 0 \\
0 & -1 & 0 & 0 & 0 & i & 0 & 0 \\
0 & -i & 0 & 0 & 0 & 1 & 0 & 0 \\
0 & 0 & -i & 0 & 0 & 0 & 1 & 0 \\
i & 0 & 0 & 0 & -1 & 0 & 0 & 0
\end{array}
\right) ; \\
\tau _{\mathcal{E}} &\equiv &\left(
\begin{array}{cc}
-3 & 4i \\
-4i & -3
\end{array}
\right) .  \label{tau-E-call}
\end{eqnarray}
On the other hand, the decomposition of the fundamental irrepr. $\mathbf{8}$
of $SU(8)$ under (\ref{U(1)_A}) reads (subscripts denote charges with
respect to $U(1)_{\mathcal{A}}$):
\begin{equation}
\begin{array}{l}
SU(8)\longrightarrow SU(2)\times SU(6)\times U(1)_{\mathcal{A}} \\
\mathbf{8}\longrightarrow \left( \mathbf{1},\mathbf{6}\right) _{1}+\left(
\mathbf{2},\mathbf{1}\right) _{-3},
\end{array}
\end{equation}
so that the generator of $U(1)_{\mathcal{A}}$ can be written as
\begin{equation}
\tau _{\mathcal{A}}\equiv \left(
\begin{array}{cc}
-3 & 0 \\
0 & -3
\end{array}
\right) .  \label{tau-A-call}
\end{equation}
Thus, the group $\Gamma $ defined by (\ref{Gamma}) is explicitly determined
by the solution of the following $2\times 2$ matrix Eq.:
\begin{equation}
\exp (is\tau _{\mathcal{A}})=\exp (is\tau _{\mathcal{E}}).  \label{cond}
\end{equation}
By using definitions (\ref{tau-E-call}) and (\ref{tau-A-call}), Eq. (\ref
{cond}) can be explicited as follows:
\begin{equation}
\mathbf{I}_{2}=\exp (-i4s\sigma _{2})=\cos (4s)1_{2\times 2}-i\sin
(4s)\sigma _{2},  \label{cond-expl}
\end{equation}
where $\sigma _{2}$ is the second Pauli matrix
\begin{equation}
\sigma _{2}\equiv \left(
\begin{array}{cc}
0 & -i \\
i & 0
\end{array}
\right) .
\end{equation}
The solution of Eq. (\ref{cond-expl}) reads
\begin{equation}
s=\kappa \frac{\pi }{2},~\kappa \in \mathbb{Z},  \label{sol-cond-expl}
\end{equation}
implying that
\begin{equation}
\Gamma =\langle \exp (-3i\kappa \frac{\pi }{2}\ \mathbf{I}_{2})\rangle
_{\kappa \in \mathbb{Z}}\equiv \mathbb{Z}_{4},  \label{Gamma-expl}
\end{equation}
as claimed.\newline

The residual symmetry $\Gamma $ (\ref{Gamma-expl}) of $\frac{1}{8}$-BPS
attractors within Adams-Iwasawa approach to $\mathcal{N}=8$, $d=4$ \textit{%
ungauged} supergravity performed above is consistent with the fact that the
charge configurations found to support the manifestly $SU\left( 8\right) $%
-invariant representative point (\ref{origin-2}) as a $\frac{1}{8}-$BPS
attractor are \textit{not} dyonic, but rather \textit{purely electric} (\ref
{el-sol}) or \textit{purely magnetic} (\ref{magn-sol}).

This has the following physical interpretation: the Iwasawa parametrization (%
\ref{coset}) is not as symmetric as the Cartan construction \cite{dWN}, but
it nevertheless remarkably provides the \textit{purely electric} / \textit{%
purely magnetic} disentangling of generically dyonic $\frac{1}{8}$-BPS
``large'' $d=4$ extremal black holes\footnote{%
However, it should be pointed out that the \textit{non-dyonic} (namely,
\textit{purely electric} or \textit{purely magnetic}) $\frac{1}{8}$-BPS
attractor solutions with $\phi _{ijkl}=0$ obtained in $d=4$ (see Eqs. (\ref
{el-sol}) and (\ref{magn-sol}), respectively) do \textit{not} necessarily
uplift to $d=5$ extremal (\textit{electric}) black holes or extremal (%
\textit{magnetic}) black strings, respectively. Indeed, the Iwasawa
symplectic basis considered in the present paper does \textit{not} coincide
with the Sezgin-Van Nieuwenhuizen \cite{SVN} symplectic frame, in which the
maximal non-compact covariance is nothing but the $d=5$ $U$-duality group $%
E_{6\left( 6\right) }$.
\par
Clearly, these two symplectic bases are related through a finite symplectic
transformation ($\in Sp\left( 56,\mathbb{R}\right) $, in the semi-classical
regime of large, continuous charges). We leave the determination of such a
transformation, along with the study of the $d=5$ uplifts of the obtained
Iwasawa \textit{non-dyonic} $d=4$ $\frac{1}{8}$-BPS ``large'' solutions, as
an issue for future investigations.}. In this sense, by recalling Eqs. (\ref
{Q1})-(\ref{Q2}) and solutions (\ref{el-sol})-(\ref{magn-sol}), the
performed intersection (\ref{Gamma}) can be interpreted as a constraint
selecting the real (or, equivalently, imaginary) part of the $\frac{1}{8}$%
-BPS ``large'' charge orbit $\mathcal{O}_{\frac{1}{8}-BPS,\text{large}}$ (%
\ref{Q-1/8-BPS-large}).\newline


\section{\label{d=4-d=5}Comments on $d=4$/$d=5$ Relations}

The present Section is devoted to consider some relations between maximal
supergravities in $d=4$ and $d=5$, at the level of scalar manifolds ($M$),
``large'' orbits ($\mathcal{O}$) and \textit{moduli spaces} of attractors ($%
\mathcal{M}$).

While many relations have been derived in
\cite{ADF-U-duality-d=4,Ferrara-Marrani-2,ICL-1}, and detailed
treatments of related issues are given in \cite{CFGn-1} and
\cite{ICL-1}, we believe that the relations to $d=3$ theories given
below can hint some further interesting developments, especially in
light of very recent and intriguing advances (see \textit{e.g.}
\cite{d=3-recent}).

\begin{itemize}
\item  The scalar manifold and the unique ``large'' charge orbit (with
related \textit{moduli space}) of $\mathcal{N}=8$, $d=5$ \textit{ungauged}
supergravity respectively read
\begin{eqnarray}
M_{\mathcal{N}=8,d=5} &=&\frac{E_{6\left( 6\right) }}{USp\left( 8\right) };
\\
\mathcal{O}_{\frac{1}{8}-BPS} &=&\frac{E_{6\left( 6\right) }}{F_{4\left(
4\right) }}; \\
\mathcal{M}_{\frac{1}{8}-BPS} &=&\frac{F_{4\left( 4\right) }}{USp\left(
6\right) \times USp\left( 2\right) }= \\
&=&M_{\mathcal{N}=4,J_{3}^{\mathbb{R}},d=3}=c\left( \frac{Sp\left( 6,\mathbb{%
R}\right) }{SU\left( 3\right) \times U\left( 1\right) }\right) ,
\label{new-rel-1}
\end{eqnarray}
with ``$c$'' denotes the $c$-map \cite{CFG}.

\item  The scalar manifold and the ``large'' charge orbits of $\mathcal{N}=8$%
, $d=4$ \textit{ungauged} supergravity are respectively given by Eqs. (\ref
{scalar-manifold}), (\ref{Q-1/8-BPS-large}) and (\ref{Q-nBPS-large}). The
related \textit{moduli spaces} of attractors respectively read
\begin{eqnarray}
\mathcal{M}_{\frac{1}{8}-BPS,\text{large}} &=&\frac{E_{6\left( 2\right) }}{%
SU\left( 6\right) \times SU\left( 2\right) }= \label{pre-new-rel-2}\\
&=&M_{\mathcal{N}=4,J_{3}^{\mathbb{C}},d=3}=c\left( \frac{SU\left(
3,3\right) }{SU\left( 3\right) \times SU\left( 3\right) \times
U\left(
1\right) }\right) ;  \label{new-rel-2} \\
\mathcal{M}_{nBPS} &=&\frac{E_{6\left( 6\right) }}{USp\left( 8\right) }=M_{%
\mathcal{N}=8,d=5}.
\end{eqnarray}
\end{itemize}

\noindent

Both $\frac{1}{8}$-BPS and non-BPS $d=4$ extremal black hole attractors are
descendants of the $\frac{1}{8}$ -BPS extremal $d=5$ (black hole \textit{or}
black string) attractors. Thus, it holds that
\begin{eqnarray}
\mathcal{M}_{\frac{1}{8}-BPS,d=5} &\subsetneq &\left( \mathcal{M}%
_{nBPS,d=4}\cap \mathcal{M}_{\frac{1}{8}-BPS,\text{large},d=4}\right) ;
\label{emb-1} \\
&\Updownarrow &  \notag \\
M_{\mathcal{N}=4,J_{3}^{\mathbb{R}},d=3} &\subsetneq &\left( M_{\mathcal{N}%
=8,d=5}\cap M_{\mathcal{N}=4,J_{3}^{\mathbb{C}},d=3}\right) ;  \label{emb-2}
\\
&\Updownarrow &  \notag \\
\frac{F_{4\left( 4\right) }}{USp\left( 6\right) \times USp\left( 2\right) }
&\subsetneq &\left( \frac{E_{6\left( 6\right) }}{USp\left( 8\right) }\cap
\frac{E_{6\left( 2\right) }}{SU\left( 6\right) \times SU\left( 2\right) }%
\right) .  \label{emb-3}
\end{eqnarray}
While (\ref{emb-1}) (or explicitly (\ref{emb-3})) has been obtained in \cite
{ICL-1}, its re-interpretation (\ref{emb-2}) is new. It involves the
quaternionic scalar manifolds of $\mathcal{N}=4$, $d=3$ ``magic''
supergravity theories based on the rank-$3$ Euclidean Jordan algebras $%
J_{3}^{\mathbb{R}}$ and $J_{3}^{\mathbb{C}}$, or equivalently in terms of
the hypermultiplets' scalar manifolds of the same theories in $d=4$. In
light of very recent developments on the timelike $d=4\rightarrow d=3$
reduction of $stu$ model (see \textit{e.g.} \cite{d=3-recent}), the
relations (\ref{new-rel-1}), (\ref{new-rel-2}) and (\ref{emb-2}) are
interesting, because $stu$ is a common sector of all symmetric rank-$3$
special K\"{a}hler geometries in $d=4$, as its $d=3$ spacelike (timelike)
reduction is a common sector of the $c$-map ($c^{\ast }$-map) images of all
symmetric rank-$3$ special K\"{a}hler geometries in $d=4$, namely of all
symmetric rank-$4$ (pseudo-)quaternionic K\"{a}hler geometries in $d=3$. In
particular:
\begin{gather}
M_{\mathcal{N}=2,stu,d=4}=\left( \frac{SL\left( 2,\mathbb{R}\right) }{%
U\left( 1\right) }\right) ^{3}  \notag \\
\downarrow c  \notag \\
M_{\mathcal{N}=4,stu,d=3}=\frac{SO\left( 4,4\right) }{SO\left( 4\right)
\times SO\left( 4\right) }\subsetneq \left\{
\begin{array}{l}
M_{\mathcal{N}=4,J_{3}^{\mathbb{R}},d=3}=\mathcal{M}_{\frac{1}{8}-BPS,%
\mathcal{N}=8,d=5}; \\
\\
M_{\mathcal{N}=4,J_{3}^{\mathbb{C}},d=3}=\mathcal{M}_{\frac{1}{8}-BPS,\text{%
large},\mathcal{N}=8,d=4}.
\end{array}
\right.
\end{gather}

\section{\label{Conclusion}Conclusion}

Through a comparison with some other approaches to $\mathcal{N}=8$, $d=4$
\textit{ungauged} supergravity, namely with the ones by Sezgin-van
Nieuwenhuizen \cite{SVN,CFGM1,CFGn-1} and Cremmer-Julia or de Wit-Nicolai
\cite{CJ,dWN,CFGM1}, the main properties of the various symplectic frames
can be summarized as follows.

\begin{enumerate}
\item  The Sezgin-van Nieuwenhuizen \cite{SVN,CFGM1,CFGn-1} construction has
$USp(8)\subsetneq _{symm}^{\max }SU(8)$ as maximal compact subgroup, which
is nothing but the maximal compact subgroup of the $\mathcal{N}=8$, $d=5$ $U$%
-duality group $E_{6\left( 6\right) }$: $USp(8)=mcs\left( E_{6\left(
6\right) }\right) $. By recalling the explicit form of ``large'' non-BPS
charge orbit in $\mathcal{N}=8$, $d=4$ supergravity \cite{FG,LPS-1}:
\begin{equation}
\mathcal{O}_{nBPS}=\frac{E_{7\left( 7\right) }}{E_{6\left( 6\right) }},
\label{Q-nBPS-large}
\end{equation}
it is easy to realize that this is the natural context in which ``large''
dyonic non-BPS $d=4$ extremal black holes can be treated \cite{CFGM1,CFGn-1}.

\item  The Cremmer-Julia or de Wit-Nicolai \cite{CJ,dWN,CFGM1}
parametrization privileges the subgroup $SO(8)=mcs\left( SL\left( 8,\mathbb{R%
}\right) \right) \subsetneq _{symm}^{\max }E_{7(7)}$, providing a natural
context in which $\frac{1}{8}-$BPS ``large'' extremal $d=4$ black holes can
be treated (in particular, in the Lie algebra approach to the scalar manifold, with all
scalars vanishing at the horizon \cite{CFGM1}).

\item  In the Iwasawa construction performed in the present paper, we
started from a realization of $E_{7(7)}$ (Adams' approach, Subsect. \ref
{Lie-Algebra}) with a natural underlying $SL(8,\mathbb{R})$ symmetry.
However, this symmetry is soon broken explicitly down to $SL(7,\mathbb{R}%
)\times SO(1,1)$ by the selection of a Cartan subalgebra (see Eq. (\ref{br-1}%
)). This is essentially due to the maximally non-compact nature of the
Cartan subalgebra $\mathcal{C}$ (\ref{C-def}) of $\frak{e}_{7\left( 7\right)
}$, used in our construction. Therefore, the maximal \textit{off-shell}
symmetry of the whole $\mathcal{N}=8$, $d=4$ Lagrangian density is actually $%
SL(7,\mathbb{R})\times SO(1,1)$, with $SO(7)$ as maximal compact subgroup
(with symmetric embedding). Since
\begin{equation}
mcs\left( SL(7,\mathbb{R})\times SO(1,1)\right) =SO(7)\subsetneq
_{symm}^{\max }SO\left( 8\right) =mcs\left( SL(8,\mathbb{R})\right) ,
\end{equation}
it is natural to expect the residual ``degeneracy'' symmetry of $\frac{1}{8}$%
-BPS attractors in the considered Iwasawa parametrization to be a proper
subgroup of $U(1)_{\mathcal{A}}$ (see Eq. (\ref{U(1)_A})). As proved in
Subsect. \ref{Breaking}, this is indeed the case, because the Adams-Iwasawa
construction performed in the present paper induces the breaking\footnote{%
It is worth mentioning that there are essentially two distinct embeddings of
$\mathbb{Z}_{4}\hookrightarrow U(1)$ given by
\begin{equation*}
\mu \mapsto \pm i,
\end{equation*}
where $\mu $ is a primitive generator of $\mathbb{Z}_{4}$.\newline
}
\begin{equation}
U\left( 1\right) _{A}\longrightarrow \mathbb{Z}_{4}.  \label{Breaking-expl}
\end{equation}
In light of previous reasonings, this breaking can also be traced back to
the ``real selection rule'' inherited from the Cayley transformation $SL(8,%
\mathbb{R})\mapsto SU(8)$ (recall (\ref{our-Cayley}) and (\ref{V-call})),
thus giving rise only to $\frac{1}{8}-$BPS ``large'' non-dyonic extremal $%
d=4 $ black holes, represented by solutions (\ref{el-sol}) and (\ref
{magn-sol}).
\end{enumerate}

The origin of $\mathbb{Z}_{4}$ of (\ref{Breaking-expl}) can also be
explained as follows. The breaking (\ref{br-1}) is clearly incompatible with
the branching
\begin{equation}
SL(8,\mathbb{R})\rightarrow SL(6,\mathbb{R})\times SL(2,\mathbb{R})\times
SO(1,1),  \label{brr}
\end{equation}
unless the $SL(7,\mathbb{R})$ in (\ref{br-1}) breaks down as
\begin{equation}
SL(7,\mathbb{R})\rightarrow SL(6,\mathbb{R})\times \Gamma ,
\end{equation}
where $\Gamma $ is some subgroup commuting with $SL(2,\mathbb{R})\times
SO(1,1)$ in (\ref{brr}). Indeed, by performing the Cayley transformation (%
\ref{our-Cayley-2}), the result (\ref{Gamma-expl}) has been achieved.

As mentioned in the Introduction, the embedding analysis of \cite{CFGM1} has
pointed out that the Lie algebra approach to $\frac{1}{8}$-BPS attractors is related
to the embedding of the Reissner-N\"{o}rdstrom extremal BH solution of
\textit{pure} $\mathcal{N}=2$, $d=4$ supergravity into $\mathcal{N}=8$
theory itself. Consequently, the $U\left( 1\right) \rightarrow \mathbb{Z}%
_{4} $ breaking (\ref{Breaking-expl}), due to the constraints on manifest
\textit{off-shell} covariance originated from the explicit Adams-Iwasawa
construction of the coset representative of $\frac{E_{7\left( 7\right) }}{%
SU\left( 8\right) }$ performed in the present paper, can be nicely
interpreted as a breaking of the global $U\left( 1\right) $ $\mathcal{R}$%
-symmetry \cite{FSZ} of the \textit{pure} $\mathcal{N}=2$, $d=4$ theory.

\section*{Acknowledgments}

We would like to thank Sergio Ferrara and Bert Van Geemen for
enlightening discussions. S. L. C. would like to thank Francesco
Dalla Piazza for providing a first version of the Adams construction
of the group $E_{7(7)}$.

B. L. C. would like to thank the University of California Berkeley
Center for Theoretical Physics, where part of this work was done,
for kind hospitality and stimulating environment.

A. M. would like to thank the Department of Mathematics,
Universit\`{a} degli Studi di Milano, Italy, where part of this work
was done, for warm hospitality and inspiring environment.

The work of B. L. C. has been supported in part by the European Commission
under the FP7-PEOPLE-IRG-2008 Grant n PIRG04-GA-2008-239412 \textit{``String
Theory and Noncommutative Geometry''} (\textit{STRING}).

The work of A. M. has been supported by an INFN visiting Theoretical
Fellowship at SITP, Stanford University, Stanford, CA, USA.

\end{document}